\documentclass[12pt]{article}
\usepackage{amsfonts}

\usepackage{mathrsfs}
\usepackage{srcltx}
\usepackage{xcolor}
\usepackage{amsmath}
\usepackage{amsthm}
\usepackage{amstext}
\usepackage{amsopn}
\usepackage{placeins}

\usepackage{subfigure}
\usepackage{tikz}

\usepackage{graphicx}
\usepackage{bm}
\newtheorem{theorem}{Theorem}[section]
\newtheorem{lemma}[theorem]{Lemma}
\newtheorem{proposition}[theorem]{Proposition}
\newtheorem{corollary}[theorem]{Corollary}

\theoremstyle{definition}

\theoremstyle{remark}
\newtheorem{remark}[theorem]{Remark}

\numberwithin{equation}{section}

\newcommand{\ba}{\begin{array}}
\newcommand{\ea}{\end{array}}
\newcommand{\f}{\frac}

\newcommand{\la}{\lambda}

\newcommand{\ds}{\displaystyle}

\usepackage{appendix}

\begin{document}
\date{}
\title{ \bf\large{On the impact of spatial heterogeneity and drift rate in a three-patch two-species Lotka-Volterra competition model over a stream}\thanks{S. Chen is supported by National Natural Science Foundation of China (Nos. 12171117, 11771109) and Shandong Provincial Natural Science Foundation of China (No. ZR2020YQ01). 
}}
\author{
Shanshan Chen\footnote{Corresponding Author, Email: chenss@hit.edu.cn}\\[-1mm]
{\small Department of Mathematics, Harbin Institute of Technology}\\[-2mm]
{\small Weihai, Shandong 264209, P. R. China}\\[1mm]
Jie Liu\footnote{Email: liujie9703@126.com}\\[-1mm]
{\small School of Mathematics, Harbin Institute of Technology}\\[-2mm]
{\small Harbin, Heilongjiang, 150001, P. R. China}\\[1mm]
Yixiang Wu\footnote{Email: yixiang.wu@mtsu.edu} \\[-1mm]
{\small Department of Mathematics, Middle Tennessee State University}\\[-2mm]
{\small Murfreesboro, Tennessee 37132, USA}\\[1mm]
}

\maketitle
\begin{abstract}
In this paper, we study a three-patch two-species Lotka-Volterra competition patch model over a stream network. 
The individuals are subject to both random and directed movements, and the two species are assumed to be identical except for the movement rates. The environment is heterogeneous, and the carrying capacity is lager in upstream locations. We treat one species as a resident species and investigate whether the other species can invade or not. Our results show that the spatial heterogeneity of environment and the magnitude of the drift rates have a large impact on the competition outcomes of the stream species. 
\\[2mm]
\end{abstract}

\section{Introduction}

The species living in stream environment is subject both passive random movement and directed drift \cite{speirs2001population}. Intuitively, the drift will carry individuals to the downstream end, which may be crowded or hostile. However, the random dispersal may drive the individuals to the upper stream locations, which are usually more favorable for the species \cite{jiang2020two}. Therefore, the joint impact of both undirectional and directed dispersal rates on the population dynamics of the species are usually complicated and have attracted increasing research interests recently \cite{HuangJin, jin2011seasonal, lou2014evolution,  lutscher2006effects, lutscher2007spatial, lutscher2005effect, speirs2001population}.

Dispersal has profound effects on the distribution and abundance of organisms, and understanding the mechanisms for the evolution of dispersal is a fundamental question related to dispersal \cite{johnson1990evolution}. In the seminal works of Hastings \cite{hastings1983can}  and Dockery \emph{et al.} \cite{dockery1998evolution}, it has been shown that in a spatial heterogeneous environment, when two competing species are identical except for the random dispersal rate, evolution of dispersal favors the species with a smaller dispersal rate.
However, in an advective environment when individuals are subject to both undirectional random dispersal and directed movement, species with a faster dispersal rate can be selected \cite{cantrell2006movement, cantrell2010evolution, chen2012dynamics}.

Two species reaction-diffusion-advection competition models of the following form has been proposed to study the evolution of dispersal for stream species \cite{lam2015evolution,lou2014evolution, lou2018coexistence, lou2016qualitative, lou2015evolution, ma2020evolution,  vasilyeva2011population, vasilyeva2012flow,  yan2022competition,zhao2016lotka, zhou2016lotka}:
\begin{equation}\label{rs-comp}
\begin{cases}
u_t=d_1u_{xx}-q_1 u_x+u[r(x)-u-v], &0<x<l,\;\;t>0,\\
v_t=d_2v_{xx}-q_2 v_x+v[r(x)-u-v], &0<x<l,\;\;t>0,\\
d_1u_x(0,t)-q_1 u(0,t)=d_2v_x(0,t)-q_2 v(0,t)=0, &t>0, \\
d_1u_x(l,t)-q_1 u(l,t)=d_2v_x(l,t)-q_2 v(l,t)=0, &t>0,\\
u(x,0), v(x, 0)\ge(\not \equiv)0,  &0<x<l.
\end{cases}
\end{equation}
In \cite{vasilyeva2011population}, the authors have studied the existence and properties of the positive equilibrium of \eqref{rs-comp}. In \cite{lam2015evolution,lou2014evolution,vasilyeva2012flow}, the authors have treated species $u$ as a resident species and studied the conditions under which the species $u$ only semitrivial equilibrium is stable/unstable. Various results on the global dynamics of \eqref{rs-comp} are presented in \cite{lou2016qualitative, lou2015evolution, ma2020evolution, zhao2016lotka, zhou2016lotka}. In particular, if $r(x)$ is constant, the works \cite{lou2016qualitative,lou2015evolution,zhou2016lotka} show that the species with a larger diffusion rate and/or a smaller advection rate wins the competition. If $r(x)$ is a decreasing function, the authors in \cite{zhao2016lotka} use $q_1$ and $q_2$ as bifurcation parameters to study the global dynamics of \eqref{rs-comp}.


To study the evolution of dispersal in a river network, the authors in \cite{jiang2020two, Jiang-Lam-Lou2021} propose and investigate three-patch two-species Lotka-Volterra competition models.  Let 
$\bm u=(u_1, u_2, u_3)$ and $\bm v=(v_1, v_2, v_3)$ be the population density of two competing species respectively, where $u_i$ and $v_i$ are the densities in patch $i$. Suppose that the dispersal patterns of the individuals and the configuration of the patches are shown in Fig. \ref{river}.
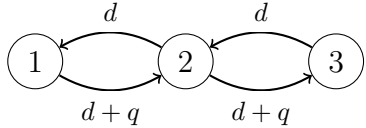
\begin{figure}[htbp]
\centering
\begin{tikzpicture}
\begin{scope}[every node/.style={draw}, node distance= 1.5 cm]
    \node[circle] (1) at (0,0) {$1$};
    \node[circle] (2) at (2,0) {$2$};
    \node[circle] (3) at (4,0) {$3$};
\end{scope}
\begin{scope}[every node/.style={fill=white},
              every edge/.style={thick}]
    \draw[thick] [->](1) to [bend right] node[below=0.1] {{\footnotesize $d+q$}} (2);
    \draw[thick] [->](2) to [bend right] node[below=0.1] {{\footnotesize $d+q$}} (3);

    \draw[thick] [<-](1) to [bend left] node[above=0.1] {{\footnotesize $d$}} (2);
    \draw[thick] [<-](2) to [bend left] node[above=0.1] {{\footnotesize $d$}} (3);
\end{scope}
\end{tikzpicture}
\caption{A stream with three patches, where $d$ is the random movement rate and $q$ is the directed drift rate. Patch 1 is  the  upstream end, and patch 3 is  the downstream end. }\label{river}
\end{figure}
The competition patch model over the stream network in Fig. \ref{river} (with $r_1=r_2=r_3$) in  \cite{jiang2020two, Jiang-Lam-Lou2021} is: 
\begin{equation}\label{pat-cp}
\begin{cases}
\ds\frac{du_i}{dt}=\ds\sum_{j=1}^{3} (d_1 D_{ij}+q_1Q_{ij})u_j+r_iu_i\left(1-\frac{u_i+v_i}{k_i}\right), &i=1,2,3,\;\;t>0,\\
\ds\frac{dv_i}{dt}=\ds\sum_{j=1}^{3}(d_2 D_{ij}+q_2Q_{ij})v_j+r_iv_i\left(1-\frac{u_i+v_i}{k_i}\right),&i=1,2,3,\;\; t>0, \\
\bm u(0)=\bm u_0\ge(\not\equiv)\;\bm0,\; \bm v(0)=\bm v_0\ge(\not\equiv)\; \bm0,
\end{cases}
\end{equation}
where $d_1$ and $d_2$ are random movement rates; $q_1$ and $q_2$ are directed movement rates;  $\bm r=(r_1, r_2, r_3)$ is the growth rate; $\bm k=(k_1, k_2, k_3)$ is the carrying capacity;  and two  $3\times 3$ matrices $D=(D_{ij})$ and $Q=(Q_{ij})$ 
 represent the random movement pattern and   directed drift pattern of individuals, respectively, where
\begin{equation}\label{DQ}
    D=\begin{bmatrix}
  -1 & 1 &0 \\
  1&-2& 1\\
  0&1 & -1 \\
 \end{bmatrix},\;\;Q=\begin{bmatrix}
  -1 & 0 &0 \\
  1&-1& 0\\
  0&1 & 0\\
 \end{bmatrix}.
\end{equation}
We can write the model as

\begin{equation}\label{3p}
\begin{cases}
\ds\frac{du_1}{dt}=-(d_1+q_1)u_1+d_1u_2+r_1u_1\left(1-\ds\f{u_1+v_1}{k_1}\right),\smallskip\\
\ds\frac{du_2}{dt}=(d_1+q_1)u_1-(2d_1+q_1)u_2+d_1u_3+r_2u_2\left(1-\ds\f{u_2+v_2}{k_2}\right),\\
\ds\frac{du_3}{dt}=(d_1+q_1)u_2-d_1u_3+r_3u_3\left(1-\ds\f{u_3+v_3}{k_3}\right),\smallskip\\
 \ds\frac{dv_1}{dt}=-(d_2+q_2)v_1+d_2v_2++r_1v_1\left(1-\ds\f{u_1+v_1}{k_1}\right),\smallskip\\
\ds\frac{dv_2}{dt}=(d_2+q_2)v_1-(2d_2+q_2)v_2+d_2v_3+r_2v_2\left(1-\ds\f{u_2+v_2}{k_2}\right),\smallskip\\
\ds\frac{dv_3}{dt}=(d_2+q_2)v_2-d_2v_3+r_3v_3\left(1-\ds\f{u_3+v_3}{k_3}\right),\\
\bm u(0)=\bm u_0\ge(\not\equiv)\;\bm0,\;\bm v(0)=\bm v_0\ge(\not\equiv)\;\bm0.
\end{cases}
\end{equation}
We assume $d_1, d_2, q_1, q_2>0$ and $r_i, k_i>0$ for $i=1, 2, 3$. We adopt the same assumption in \cite{jiang2020two}  on $\bm k=(k_1, k_2, k_3)$:
\begin{enumerate}
\item [$(\bf{H})$]  $k_1> k_2 > k_3>0$. 
\end{enumerate}
Biologically, $(\bf{}{H})$ means that the upstream locations are more favorable for both species. Model \eqref{3p} has two semitrivial equilibria $(\bm u^*,\bm 0)$ and $(\bm 0,\bm v^*)$.

Two-species Lotka-Volterra competition patch models have attracted many research interests recently. Model \eqref{pat-cp} with $n$ patches in spatially homogeneous environment (i.e. $r_1=\dots=r_n$ and $k_1=\dots=k_n$) has been considered in our earlier papers \cite{chen2022invasion,chen2021}, but many techniques and results there cannot be generalized to the situation when $\bm k$ is non-constant. 
The authors in \cite{hamida2017evolution,noble2015evolution} have studied the global dynamics of model \eqref{pat-cp} with two patches and $q:=q_1=q_2$. They have showed that there exists a critical drift rate such that below it the species with a smaller dispersal rate wins the competition while above it the species with a larger dispersal rate wins. 
In a competition model with two patches, the authors in \cite{cheng2019coexistence,gourley2005two,lin2014global} have showed that the species with more evenly distributed resources has less competition advantage. In \cite{chen2022global},  the global dynamics of a Lotka-Volterra competition patch model is classified under some assumptions on patches, which requires $d_1/q_1=d_2/q_2$ in terms of \eqref{pat-cp}.
For more studies on competition patch models, we refer to the works \cite{stephen2007ideal, cantrell2012evolutionary, cantrell2017evolution, kirkland2006evolution, levin1984dispersal,lou2019ideal, mcpeek1992evolution, smith2008monotone, xiang2019evolutionarily}.


We will take an adaptive dynamics approach \cite{dieckmann1996dynamical, geritz1998evolutionarily} to analyze \eqref{3p} by viewing species $\bm u$ as the resident species and species $\bm v$ as the mutant/invading species. We fix parameters $d_1$ and $q_1$ and vary $d_2$ and $q_2$. We show that there exists a curve $q=q_{\bm u}^*(d)$ dividing the $(d_2, q_2)$-plane into two regions such that $(\bm u^*, \bm 0)$ is  stable if and only if $(d_2, q_2)$ is above the curve.
Our results complement  those in \cite{jiang2020two} by defining and analyzing \ the curve  $q=q_{\bm u}^*(d)$ and obtaining the global dynamics of model \eqref{3p}. In particular, we show that if $q_1<\underline q$  the curve $q=q_{\bm u}^*(d)$ is bounded (see Fig. \ref{small}) and if $q_1>\overline q$ it is unbounded (see Fig. \ref{large}). This result is in sharp contrast with the corresponding one for the model in spatially homogeneous environment ($k_1=k_2=k_3$) \cite{chen2022invasion}, where the  curve $q=q_{\bm u}^*(d)$ is always unbounded. We give explicitly parameter ranges for competitive exclusion and conditions for coexistence/bistability  in three cases ($q_1<\underline q$, $\underline q\le q_1\le\overline q$ and $q_1>\overline q$). Our results show that the magnitude of  the drift rates and the spatial heterogeneity of  environment have an large impact on the competition outcomes of the stream species.

Our paper is organized as follows. In section 2, we list some preliminary results. In section 3, we state the main results on model \eqref{3p}. We give some conclusive remarks and numerical simulations in section 4. The proofs of the main results are presented in section 5. In the appendix, we show the relations of $\underline q$, $\overline q$, and $q_0$. These relations are implicitly included in the main results, and we prove them for reader's convenience. 

 \section{Preliminary}\label{section_pre}
Let $A=(a_{ij})_{n\times n}$ be a square matrix with real entries, $\sigma(A)$ be the set of all eigenvalues of $A$, and $s(A)$ be the  {\it spectral bound} of $A$, i.e.
$s(A)=\max\{{\rm Re} \lambda: \lambda\in\sigma(A)\}$.
The matrix $A$ is called \emph{irreducible} if it cannot be placed into block upper triangular form by simultaneous row and column permutations and \emph{essentially nonnegative} if $a_{ij}\ge 0$ for all $1\le i, j\le n$ and $i\neq j$. By the Perron-Frobenius Theorem, if $A$ is irreducible and essentially nonnegative, then $s(A)$ is an eigenvalue of $A$ (called the \emph{principal eigenvalue} of $A$), which is the unique eigenvalue associated with a nonnegative eigenvector. The following result on the monotonicity of  spectral bound can be found in \cite{altenberg2012resolvent, chen2022two}:
\begin{lemma}\label{theorem_quasi}
Let $A=(a_{ij})_{n\times n}$ be an irreducible  and essentially nonnegative  matrix and  $M=\text{diag}(m_i)$ be a real diagonal  matrix.   If $s(A)=0$, then
$$
\ds\frac{d}{d\mu} s(\mu A+M)\le 0
$$
for $\mu\in (0, \infty)$ and the  inequality is strict except for the case $m_1=\cdots=m_n$. Moreover,
$$
\lim_{\mu\rightarrow 0}s(\mu A+M)=\max_{1\le i\le n}\{m_i\} \;\;
\text{and} \;\; \lim_{\mu\rightarrow\infty}s(\mu A+R)=\sum_{i=1}^n{\theta_im_i},
$$
where $\theta_i\in (0, 1)$, $1\le i\le n$, is determined by $A$ and $\ds\sum_{i=1}^n{\theta_i}=1$ (if $A$ has each column sum equaling zero, then $\bm\theta=(\theta_1,\dots,\theta_n)^T$ is a positive eigenvector of $A$ corresponding to eigenvalue $0$).
\end{lemma}

Let  $\bm m=(m_1, m_2, m_3)$ be a real vector. We write $\bm m\gg \bm 0$ if $m_i>0$ for all $i=1, 2, 3$, and $\bm m>\bm 0$ if $\bm m\ge \bm 0$ and $\bm m\neq \bm 0$. 
Matrix $dD+qQ+\text{diag}(m_i)$ is irreducible and essentially nonnegative for any $d, q>0$, where $D$ and $Q$ are defined by \eqref{DQ}. By the Perron-Frobenius Theorem, $s\left(dD+qQ+\text{diag}(m_i)\right)$ is the principal eigenvalue of the following eigenvalue problem:
\begin{equation}\label{eigen}
\ds\sum_{j=1}^{3}(dD_{ij}+qQ_{ij})\phi_j+m_i\phi_i=\la\phi_i, \;\;i=1, 2, 3.
\end{equation}

We need to consider the following single species patch model:
\begin{equation}\label{pat-s}
\begin{cases}
\ds\frac{du_i}{dt}=\ds\sum_{j=1}^{3}(dD_{ij}+qQ_{ij})u_j+r_iu_i\left(1-\frac{u_i}{k_i}\right),&i=1,2,3,\;\;t>0,\\
\bm u(0)=\bm u_0>\bm0.
\end{cases}
\end{equation}
The global dynamics of \eqref{pat-s} is as follows: 
\begin{lemma}\label{DS-single}
Let $D$ and $Q$ be defined in \eqref{DQ}, $\bm r, \bm k\gg\bm 0$,  $d>0$, and $q\ge 0$. Then model \eqref{pat-s} admits a unique positive equilibrium $\bm u^*\gg\bm 0$, which is globally asymptotically stable.
\end{lemma}
\begin{proof}
By \cite{cosner1996variability,li2010global,Lu1993},  it suffices to show that $\bm 0$ is unstable, i.e. 
$$
s:=s\left(dD+qQ+\text{diag}(r_i)\right)>0.
$$
Let $\phi^T=(\phi_1, \phi_2, \phi_3)^T\gg\bm 0$ with  $\sum_{i=1}^3\phi_i=1$ be the positive eigenvector of  $dD+qQ+\text{diag}(r_i)$ corresponding to $s$. Multiplying $(1, 1, 1)$ to the left of $dD\phi+qQ\phi+\text{diag}({r_i})\phi=s\phi$, we get $s=\sum_{i=1}^3{{r_i\phi_i}}>0$. This proves the result. 
\end{proof}

By Lemma \ref{DS-single}, model \eqref{3p} has two semitrivial equilibria $(\bm u^*, \bm 0)$ and $(\bm 0, \bm v^*)$, where $\bm u^* (\text{resp., }\bm v^*)\gg \bm 0$ is the positive equilibrium of \eqref{pat-s} with $(d, q)$ replaced by $(d_1, q_1)$ (resp., $(d_2, q_2)$). Linearizing model \eqref{3p} at  $(\bm u^*,\bm 0)$, we can easily see that its stability is determined by the sign of  $\la_1\left(d_2,q_2,\bm {1}-{\bm u^*}/{\bm k}\right)$, which is the principal eigenvalue of  the following eigenvalue problem:
\begin{equation}\label{3eigen}
\begin{cases}
\la\phi_1=-(d_2+q_2)\phi_1+d_2\phi_2+r_1\left(1-\ds\frac{u^*_1}{k_1}\right)\phi_1,\\
\la\phi_2=(d_2+q_2)\phi_1-(2d_2+q_2)\phi_2+d_2\phi_3+r_2\left(1-\ds\frac{u^*_2}{k_2}\right)\phi_2,\\
\la\phi_3=(d_2+q_2)\phi_2-d_2\phi_3+r_3\left(1-\ds\frac{u^*_3}{k_3}\right)\phi_3.
\end{cases}
\end{equation}
In particular, $(\bm u^*,\bm 0)$ is locally asymptotically stable if $\la_1\left(d_2,q_2,\bm {1}-{\bm u^*}/{\bm k}\right)<0$ and  unstable if $\la_1\left(d_2,q_2,\bm {1}-{\bm u^*}/{\bm k}\right)>0$. Here, we abuse the notation by denoting $\bm {1}-{\bm u^*}/{\bm k}:=(1-u_1^*/k_1, 1-u_2^*/k_2, 1-u_3^*/k_3)$.

\section{Main results}

We fix $d_1, q_1>0, \bm r, \bm k\gg\bm 0$ and view species $\bm u$ as the resident species and $\bm v$ as the invading/mutating species.   We investigate the dynamics of  model \eqref{3p} varying $(d_2, q_2)$. For this purpose, we divide the first quadrant of the $(d, q)$-plane into six regions:
\begin{equation}\label{GGG}
\begin{split}
G_{11}:=&\{(d,q):\; d\ge d_1,q\ge  \frac{q_1}{d_1}d,(d,q)\ne(d_1,q_1)\},\\
G_{12}:=&\{(d,q):\; 0<d<d_1,q\ge q_1\},\\
G_{13}:=&\{(d, q):\; d\ge d_1,q_1\le q< \frac{q_1}{d_1}d,(d,q)\ne(d_1,q_1)\},\\
G_{21}:=&\{(d,q):\; 0<d \le d_1,0<q\le \ds\f{q_1}{d_1}d,(d,q)\ne(d_1,q_1)\},\\
G_{22}:=&\{(d,q):\; d>d_1,0<q\le q_1\},\\
G_{23}:=&\{(d,q):\; 0<d\le d_1,\f{q_1}{d_1}d<q\le q_1,(d,q)\ne(d_1,q_1)\}.
\end{split}
\end{equation}
\begin{figure}[ht]
\centering\includegraphics[width=0.6\textwidth]{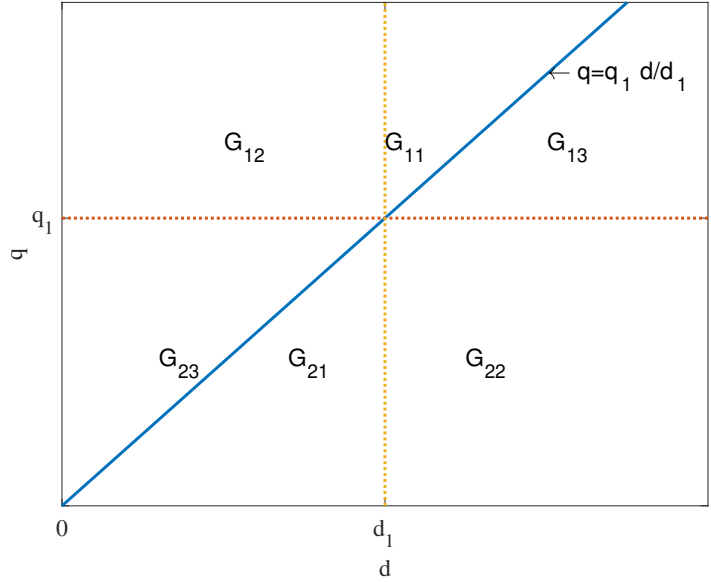}\\
\caption{Illustration of the six regions of $(d, q)$-plane.}
\label{fig1}
\end{figure}
For readers' convenience, we graph the six regions in Fig. \ref{fig1}. 

\subsection{Invasion curve}
 We consider the local stability of $(\bm u^*,\bm 0)$ in this subsection. Biologically, if $(\bm u^*,\bm 0)$ is stable, then a small amount of species $\bm v$ cannot  invade species $\bm u$; if $(\bm u^*,\bm 0)$ is unstable, then a small amount of species $\bm v$ may be able to invade species $\bm u$. We prove that there exists a curve $q=q_{\bm u}^*(d)$ in the $(d, q)-$plane such that $(\bm u^*,\bm 0)$ is locally asymptotically stable if $(d_2, q_2)$ is above the curve and $(\bm u^*,\bm 0)$ is unstable if it is below the curve. To this end, we define 
\begin{equation}\label{dstar}
d^*=\begin{cases}
\infty,&\text{if}\;\;\sum_{i=1}^3 r_i\left(1-\frac{u_i^*}{k_i}\right)\ge 0,\\
d_0,&\text{if}\;\; \sum_{i=1}^3 r_i\left(1-\frac{u_i^*}{k_i}\right)< 0,
\end{cases}
\end{equation}
where $d=d_0>0$ is the unique root of $\la_1(d, 0, \bm 1-\bm {u^*/k})=0$ if  $\sum_{i=1}^3 r_i\left(1-{u_i^*}/{k_i}\right)< 0$ (see the existence of $d_0$ in  Lemma \ref{eigenv}).
We have the following result about the local stability/instability of the semitrivial equilibrium $(\bm u^*,\bm 0)$:
\begin{theorem}\label{inva}
Suppose that $(\bf {H})$ holds, $\bm r\gg\bm 0$, and $d_1,q_1>0$. Then there exists a continuous function $q=q_{\bm u}^*(d): (0, d^*)\to \mathbb{R}_+$ passing through $(d_1, q_1)$ such that the following statements hold for model \eqref{3p}:
\begin{enumerate}
\item [{$\rm (i)$}] If $(d_2, q_2)\in S_1$, then the semitrivial equilibrium $(\bm u^*,\bm 0)$  is locally asymptotically stable;
\item [{$\rm (ii)$}] If $(d_2, q_2)\in S_2$, then the semitrivial equilibrium $(\bm u^*,\bm 0)$  is unstable.
\end{enumerate}
Here, $S_1\cup S_2$ is a partition of the first quadrant of the $(d, q)$-plane defined as follows:
\begin{equation}\label{S}
\begin{split}
S_1:=&\left\{(d,q):\; 0<d<d^*,\;\;q>q_{\bm u}^*(d)\right\}\cup S_1^*,\\
S_2:=&\left\{(d,q):\;0<d<d^*,\;0<q<q_{\bm u}^*(d)\right\},
\end{split}
\end{equation}
where
\begin{equation}
    S^*_1=\begin{cases}
    \left\{(d,q):d\ge d^*,\;\;q>0\right\}, &\;\;\text{if}\;\;d^*\ne\infty,\\
    \emptyset, &\;\;\text{if}\;\;d^*=\infty.\\
    \end{cases}
\end{equation}
\end{theorem}

\begin{remark}
We call the curve in the first quadrant of  $(d, q)$-plane defined by the function $q=q_{\bm u}^*(d)$ in Theorem \ref{inva} the \emph{invasion curve}. This curve consists with all the points $(d, q_{\bm u}^*(d))$ such that $\lambda_1(d, q_{\bm u}^*(d), \bm {1}-{\bm u^*}/{\bm k})=0$, i.e., $(\bm u^*,\bm 0)$  is linearly neutrally stable.   The invasion curve divides the first quadrant into $S_1\cup S_2$, where $S_1$ is the region above the curve and $S_2$ is the region below it. By Theorem \ref{inva}, $(\bm u^*,\bm 0)$  is locally asymptotically stable in $S_1$ and unstable in $S_2$.
\end{remark}

In the following of this paper, we denote
\begin{equation}\label{larsdr}
\begin{split}
&\underline q:=\min\left\{\ds\f{r_1}{k_1}(k_1-k_2),\ds\f{r_3}{k_3}(k_2-k_3)\right\},\\
&\overline q:=\max\left\{\ds\f{r_1}{k_1}(k_1-k_2),\ds\f{r_3}{k_3}(k_2-k_3)\right\}.
\end{split}
\end{equation}
We take $\underline q$ and $\overline q$ as the threshold values for the drift rates. Specifically, if a drift rate is below  $\underline q$ (above $\overline q$), we call it a \emph{slow (large) drift}; if a drift rate is between $\underline q$ and $\overline q$, we call it an \emph{intermediate drift}. These definitions coincide with those in \cite{Jiang-Lam-Lou2021} if $r_1=r_2=r_3$. It turns out that the magnitude of  drift rate $q_1$ will have a large impact on the shape of the invasion curve and the dynamics of the model. 

We have the following result about the invasion curve:
\begin{proposition}\label{locs}
Suppose that $(\bf {H})$ holds, $\bm r\gg\bm 0$, and $d_1,q_1>0$. Let $S_1$ and $S_2$ be defined in Theorem \ref{inva}. Then the following statements hold:
\begin{enumerate}
    \item [${\rm (i)}$]  $ G_{11}\subset S_1$ and $G_{21}\subset S_2$;
     \item [${\rm (ii)}$] If $q_1>\overline q$, then $G_{12}\subset S_1$ and
     $G_{22}\subset S_2$;
     \item [${\rm (iii)}$] If $q_1<\underline q$, then $G_{13}\subset S_1$
     and $G_{23}\subset S_2$
\end{enumerate}
\end{proposition}

 We explore further properties of the invasion
curve:
\begin{proposition}\label{profi}
Suppose that $(\bf {H})$ holds, $\bm r\gg\bm 0$, and $d_1,q_1>0$. Let $q=q_{\bm u}^*(d): (0, d^*)\to \mathbb{R}_+$ be defined in Theorem \ref{inva}. Then the following statements hold:
\begin{enumerate}
    \item [${\rm (i)}$] $\lim_{d\to0}q_{\bm u}^*(d)=q_0$, where
\begin{equation}\label{qinf}
   q_0=\max\left\{r_1\left(1-\frac{u_1^*}{k_1}\right),\ r_2\left(1-\frac{u_2^*}{k_2}\right)\right\};
\end{equation}
\item [${\rm (ii)}$] If $q_1<\underline q$, then
\begin{equation}\label{qinf1}
  d^*=d_0\;\;\text{and}\;\; \lim_{d\to d^*} q_{\bm u}^*(d)=0;
\end{equation}

\item [${\rm (iii)}$] If $q_1>\overline q$, then 
\begin{equation}\label{qinf2}
  d^*=\infty\;\;\text{and}\;\;\lim_{d\to\infty} \ds\frac{q_{\bm u}^*(d)}{d}=\theta\;\;\text{for some}\;\;\theta\in\left(0,\ \frac{q_1}{d_1}\right);
\end{equation}
\item [${\rm (iv)}$] If $\underline q\le q\le \overline q$, then \eqref{qinf1} holds when $\sum_{i=1}^3 r_i\left(1-{u_i^*}/{k_i}\right)< 0$, and 
\eqref{qinf2} holds when $\sum_{i=1}^3 r_i\left(1-{u_i^*}/{k_i}\right)\ge 0$.
\end{enumerate}
\end{proposition}

\begin{remark}
By Propositions \ref{locs}-\ref{profi}, the invasion curve lies in $G_{12}\cup G_{22}$ when the drift rate $q_1$ is small, and it lies in $G_{13}\cup G_{23}$ when $q_1$ is large. Moreover, if $q_1$ is small, the invasion curve is defined on a bounded interval $(0, d_0)$; if $q_1$ is large, it is defined on $(0, \infty)$ and has a slant asymptote $q=\theta d$ for some $\theta\in (0, q_1/d_1)$. 
\end{remark}

\subsection{Competitive exclusion}
In this subsection, we study the global dynamics of model \eqref{3p} and find some parameter ranges of  $(d_2, q_2)$ such that competitive exclusion happens. The relations of $\underline q$, $\overline q$ and $q_0$ are implicitly included in the results below. However, for reader's convenience, we include the proof in the appendix. 

Firstly, we consider the small drift case, i.e., $q_1<\underline q$.
\begin{theorem}\label{gdyn}
Suppose that $(\bf {H})$ holds, $\bm r\gg\bm 0$, and $d_1,q_1>0$ with $q_1<\underline q$. Then the following statements hold:
\begin{enumerate}
    \item [${\rm (i)}$] If $(d_2,q_2)\in G_{21}\cup G_{23}$, then the semi-trivial equilibrium $(\bm 0,\bm v^*)$ of \eqref{3p} is globally asymptotically stable;
    \item [${\rm (ii)}$] If $(d_2,q_2)\in G_{11}\cup G_{12}^*\cup G_{13}$,  then
    the semi-trivial equilibrium $(\bm u^*,\bm 0)$ of \eqref{3p} is globally asymptotically stable.
\end{enumerate}
Here, $G_{12}^*$ is defined by 
 \begin{equation}\label{G1213s}
 G^*_{12}=\{(d_2,q_2): (d_2,q_2)\in G_{12},q_2>\overline q\}.
    \end{equation}
\end{theorem}

\begin{figure}[ht]
\centering\includegraphics[width=0.8\textwidth]{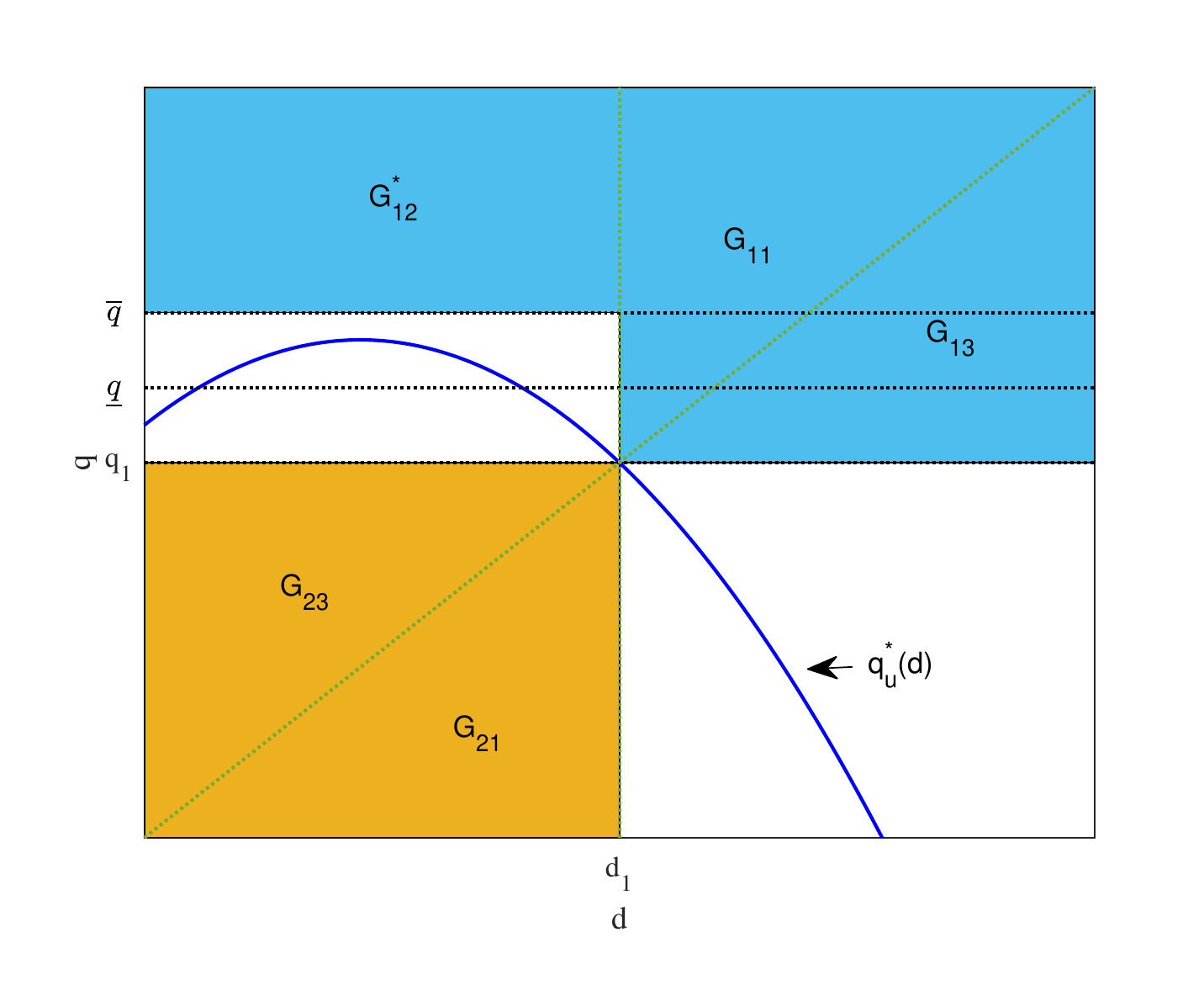}\\
\caption{Illustration of the result for the case $q_1<\underline{q}$.  If $(d_2, q_2)$ is above the curve $q=q_{\bm{u}}^*(d)$, then  $(\bm u^*,\bm 0)$ is stable; and if $(d_2, q_2)$ is under the curve, then $(\bm u^*,\bm 0)$ is unstable. If $(d_2,q_2)\in G_{21}\cup G_{23}$, $(\bm0,\bm v^*)$ is globally asymptotically stable; if $(d_2,q_2)\in G_{11}\cup G_{12}^*\cup G_{13}$, $(\bm u^*,\bm 0)$ is globally asymptotically stable. 
}
\label{small}
\end{figure}

\begin{remark}
Our results on model \eqref{3p} for the small drift rate case are summarized in Fig. \ref{small}. We have proved that competitive exclusion appears if $(d_2, q_2)$ falls into the blue and yellow regions of Fig. \ref{small}. 
\end{remark}

Next, we consider the large drift case, i.e.,  $q_1>\overline q$.
\begin{theorem}\label{gdyn1}
Suppose that $(\bf {H})$ holds, $\bm r\gg\bm 0$, and $d_1,q_1>0$ with $q_1>\overline q$. Then the following statements hold:
\begin{enumerate}
    \item [${\rm (i)}$] If $(d_2,q_2)\in G_{21}\cup G_{22}\cup G_{23}^*$, then the semi-trivial equilibrium $(\bm 0,\bm v^*)$  is globally asymptotically stable;
    \item [${\rm (ii)}$] If $(d_2,q_2)\in G_{11}\cup G_{12}$,  then the semi-trivial equilibrium $(\bm u^*,\bm 0)$  is globally asymptotically stable.
\end{enumerate}
Here, $G_{23}^*$ is defined by
\begin{equation}\label{G23*0}
G_{23}^*=\{(d_2,q_2):(d_2,q_2)\in G_{23},q_2<\underline q\}.
\end{equation}
\end{theorem}

\begin{figure}[ht]
\centering\includegraphics[width=0.8\textwidth]{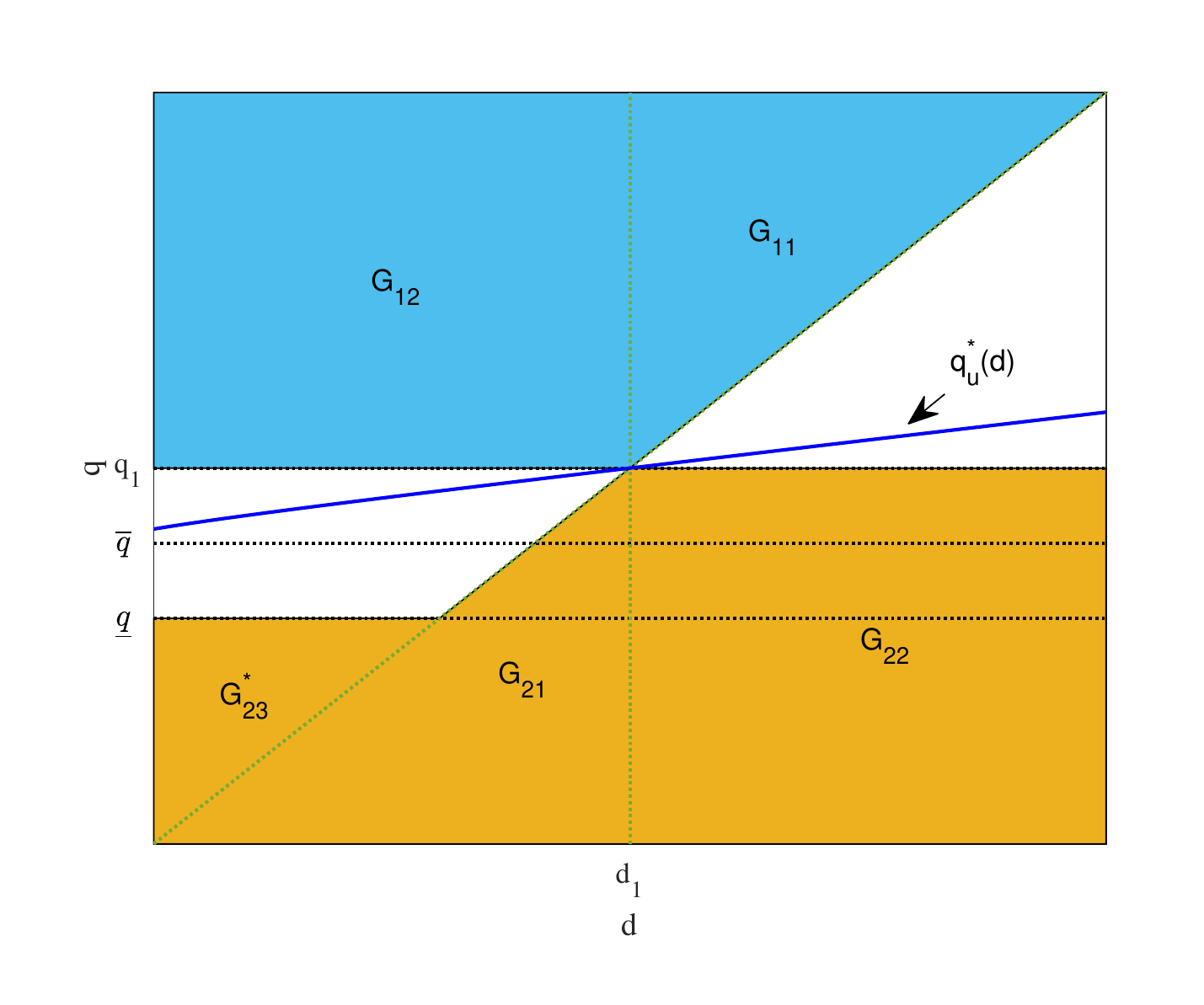}\\
\caption{Illustration of the result for the case $q_1>\overline{q}$. If $(d_2, q_2)$ is above the curve $q=q_{\bm{u}}^*(d)$, then  $(\bm u^*,\bm 0)$ is stable; and if $(d_2, q_2)$ is under the curve, then $(\bm u^*,\bm 0)$ is unstable. If $(d_2,q_2)\in G_{21}\cup G_{22}\cup G_{23}^*$, $(\bm0,\bm v^*)$ is globally asymptotically stable; and if $(d_2,q_2)\in G_{11}\cup G_{12}$, $(\bm u^*,\bm 0)$ is globally asymptotically stable.}
\label{large}
\end{figure}

\begin{remark}
Our results on model \eqref{3p} for the large drift rate case are summarized in Fig. \ref{large}. Different from the small drift rate case, the invasion curve is unbounded. Again, we are able to prove that competitive exclusion happens if $(d_2, q_2)$ falls into the blue and yellow regions of Fig. \ref{large}.
\end{remark}

Then, we consider the intermediate drift case, i.e., $\underline q\le q_1\le \overline q$, and we have the following result on the global dynamics of model \eqref{3p}.
\begin{theorem}\label{gdyn2}
Suppose that $(\bf {H})$ holds, $\bm r\gg\bm 0$, and $d_1,q_1>0$ with $\underline q\le q_1\le \overline q$. Let $G_{12}^*$ be defined by \eqref{G1213s} and $G_{23}^*$ be defined by \eqref{G23*0}. Then the following statements hold:
\begin{enumerate}
    \item [${\rm (i)}$] If $(d_2,q_2)\in G_{21}\cup G_{23}^*$, then the semi-trivial equilibrium $(\bm 0,\bm v^*)$  is globally asymptotically stable;
    \item [${\rm (ii)}$] If $(d_2,q_2)\in G_{11}\cup G^*_{12}$, then
    the semi-trivial equilibrium $(\bm u^*,\bm 0)$  is globally asymptotically stable.
\end{enumerate}
\end{theorem}

\begin{figure}[ht]
\centering\includegraphics[width=0.8\textwidth]{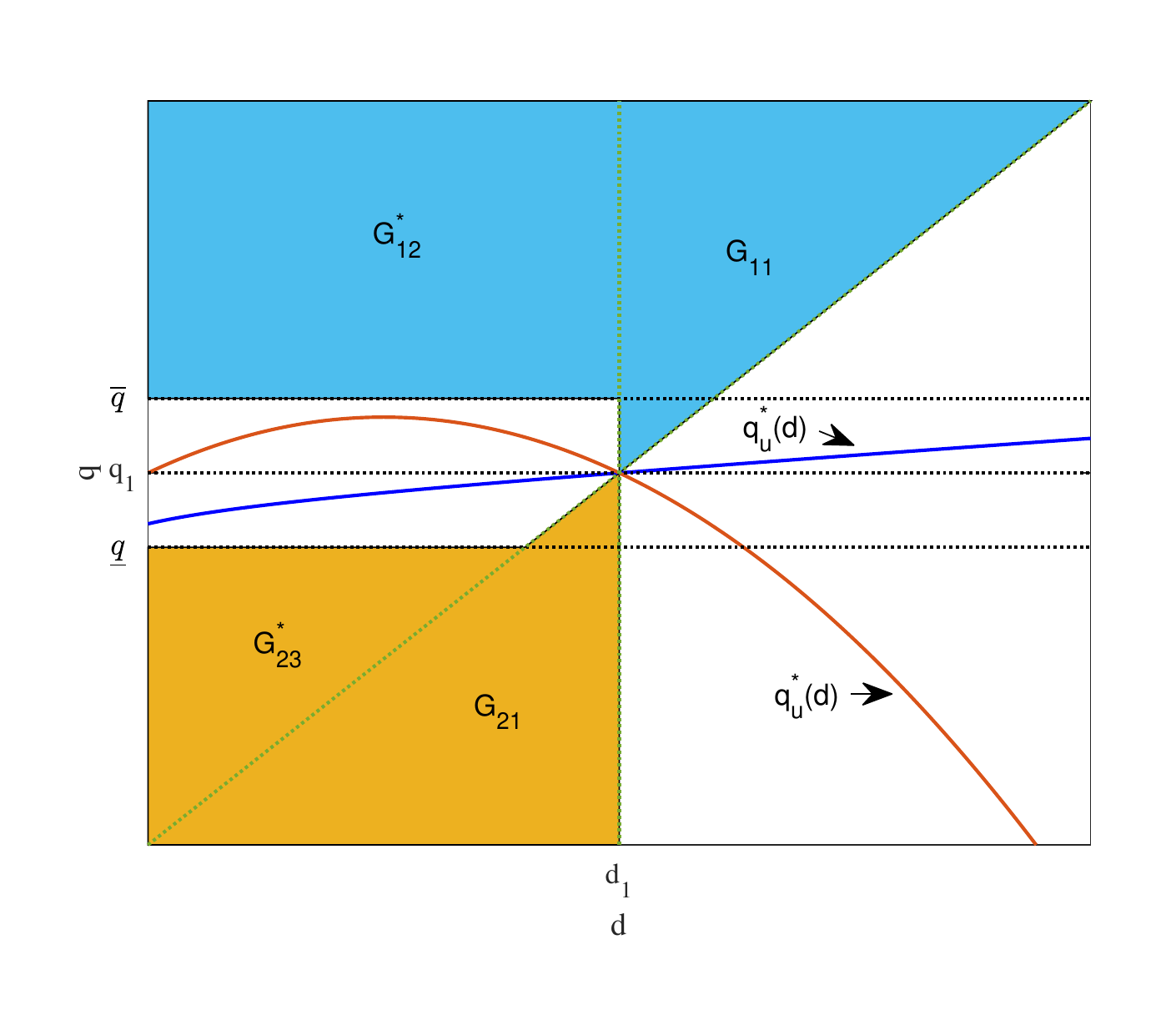}\\
\caption{Illustration of the result for the case $\underline{q}\le q_1\le \overline{q}$. If $(d_2,q_2)\in G_{21}\cup G_{23}^*$, $(\bm0,\bm v^*)$ is globally asymptotically stable; and if $(d_2,q_2)\in G_{11}\cup G^*_{12}$, $(\bm u^*,\bm 0)$ is globally asymptotically stable. }
\label{inter}
\end{figure}

\begin{remark}
Our results on model \eqref{3p} for the intermediate drift rate case are summarized in Fig. \ref{inter}. In this case, the invasion curve may be defined on either a bounded or an unbounded interval.  However, we know that it must locate between the yellow and blue regions in Fig. \ref{inter}, where competitive exclusion happens. 
\end{remark}

In view of Theorems \ref{gdyn}, \ref{gdyn1}, and \ref{gdyn2}, the global dynamics of model \eqref{3p} in $G_{11}\cup G_{21}$ is independent of $q_1$: 
\begin{corollary}\label{corollary_global3p}
Suppose that $(\bf {H})$ holds, $\bm r\gg\bm 0$, and $d_1,q_1>0$. Then the following statements hold:
\begin{enumerate}
\item [${\rm (i)}$] If $(d_2,q_2)\in G_{11}$, then  the semitrivial equilibrium $(\bm u^*,\bm 0)$  is globally asymptotically stable; 

\item [${\rm (ii)}$] If $(d_2,q_2)\in G_{21}$, then  the semitrivial equilibrium   $(\bm 0,\bm v^*)$  is globally asymptotically stable.
\end{enumerate}
\end{corollary}

More importantly, we have the following result about the evolution of random dispersal and directed drift rates.

\begin{corollary}\label{corollary_global3p1}
Suppose that $(\bf {H})$ holds, $\bm r\gg\bm 0$, and $d_1,q_1>0$. Then the following statements hold:
\begin{enumerate}
\item [${\rm (i)}$] Fix $d_1=d_2$. If $q_1<q_2$, then   the semitrivial equilibrium  $(\bm u^*,\bm 0)$  is globally asymptotically stable; If $q_1>q_2$, then  the semitrivial equilibrium  $(\bm 0,\bm v^*)$  is globally asymptotically stable;

\item [${\rm (ii)}$] Fix $q_1=q_2<\underline q$. If $d_1<d_2$, then   the semitrivial equilibrium  $(\bm u^*,\bm 0)$   is globally asymptotically stable; If $d_1>d_2$, then  the semitrivial equilibrium  $(\bm 0, \bm v^*)$   is globally asymptotically stable;

\item [${\rm (ii)}$] Fix $q_1=q_2>\overline q$. If $d_1<d_2$, then  the semitrivial equilibrium  $(\bm 0, \bm v^*)$  is globally asymptotically stable; If  $d_1>d_2$, then   the semitrivial equilibrium  $(\bm u^*, \bm 0)$  is globally asymptotically stable.
\end{enumerate}
\end{corollary}

\begin{remark}
By Corollary \ref{corollary_global3p1}, the species with a smaller drift rate always has competitive advantage. If the drift rate is small, the species with smaller random dispersal rate has competitive advantage; if the drift rate is large,  larger random dispersal rate is selected. 
\end{remark}

\subsection{Coexistence and bistability}
If $(d_2, q_2)$ is in the blank regions of Fig. \ref{small}-\ref{inter},  we show that bistability and coexistence may occur. To this end, we explore the stability/instability of the semi-trivial equilibrium $(\bm 0,\bm v^*(d_2,q_2))$  along the invasion curve, i.e. $q_2=q_{u}^*(d_2)$.
Let
\begin{equation}\label{la1lal}
\hat\la_1(d_2):=\la_1\left(d_1,q_1,\bm 1-\frac{{\bm v^*}\left(d_2,q_{\bm u}^*(d_2)\right)}{\bm k}\right).
\end{equation}
 Then  $\hat\la_1(d_1)=0$,  the semi-trivial equilibrium $(\bm 0,\bm v^*(d_2,q^*_{\bm {u}}(d_2)))$ is stable if
$\hat\la_1(d_2)<0$, and $(\bm 0,\bm v^*(d_2,q^*_{\bm {u}}(d_2)))$ is unstable if $\hat\la_1(d_2)>0$.
The following result for the large drift case can be proved similarly as \cite[Theorem 5.4]{chen2022invasion}, so we omit the proof here.
\begin{theorem}\label{theorem_cobil}
Suppose that $(\bf {H})$ holds, $\bm r\gg\bm 0$, and $d_1,q_1>0$ with $q_1>\overline q$. Let $q=q_{\bm u}^*(d): (0, d^*)\to \mathbb{R}_+$ be defined in Theorem \ref{inva} with $d^*=\infty$ (see Proposition \ref{profi} (iii)).  Then for any $d_2>0$, the following statements holds:
\begin{enumerate}
    \item [{$\rm (i)$}] If $ \hat\la_1(d_2)<0$, then
    \begin{equation*}
       \hat q(d_2) :=\inf\left\{q>0: \ q>q_{\bm u}^*(d_2)\ \text{and} \  \la_1\left(d_1,q_1,\bm 1-\frac{{\bm v^*}\left(d_2,q\right)}{\bm k}\right)\ge 0\right\}
    \end{equation*}
    exists and satisfies 
     \begin{equation}\label{estima2l}
        \begin{cases}
           \hat q(d_2)\in \left( q_{\bm u}^*(d_2),q_1 \right)\;\;&\text{for}\;\; d_2<d_1,\\
           \hat q(d_2)\in \left( q_{\bm u}^*(d_2),\ds\f{q_1}{d_1}d_2\right)\;\;&\text{for}\;\; d_2>d_1.
        \end{cases}
    \end{equation}
    Moreover, for any $q_2\in(q_{\bm u}^*(d_2), \hat q(d_2))$, both semitrivial equilibria $(\bm u^*,\bm 0)$ and $(\bm 0, \bm v^*)$ are locally asymptotically stable and  model \eqref{3p} admits an unstable positive equilibrium.

  \item [{$\rm (ii)$}] If $ \hat\la_1(d_2)>0$, then
    \begin{equation*}
       \hat q(d_2) :=\sup\left\{q>0:\  q<q_{\bm u}^*(d_2)\ \text{and}\  \la_1\left(d_1,q_1,\bm 1-\frac{{\bm v^*}\left(d_2,q\right)}{\bm k}\right)\le 0\right\}
    \end{equation*}
    exists and satisfies
        \begin{equation*}
        \begin{cases}
           \hat q(d_2)\in \left(\ds\f{q_1}{d_1}d_2, q_{\bm u}^*(d_2)\right)\;\; &\text{for}\;\; d_2<d_1,\\
           \hat q(d_2)\in \left(q_1, q_{\bm u}^*(d_2)\right)\;\;&\text{for}\;\; d_2>d_1.
        \end{cases}
    \end{equation*}
    Moreover, for any $q_2\in (\hat q(d_2),q_{\bm u}^*(d_2))$, both  semitrivial equilibria $(\bm u^*,\bm 0)$ and $(\bm 0, \bm v^*)$ are unstable and  model \eqref{3p} admits a stable positive equilibrium.
\end{enumerate}
\end{theorem}
\begin{remark}
If $\underline q\le q_1\le \overline q$ (the intermediate drift case), Theorem \ref{theorem_cobil} (i)-(ii) holds for any $d_2<d_1$, and  we omit the statement to save space here.
\end{remark}

The small drift rate case will be handled sightly different from the large drift rate case. 
For any $\theta>0$, by Lemma \ref{theorem_quasi} and Proposition \ref{locs} (ii), the line $q=d\theta$ and the invasion curve $q=q_{\bm u}^*(d)$ have exactly one intersection point $(d^*(\theta), d^*(\theta)\theta)$. So we can reparameterize the invasion curve  as follows:
\begin{equation}\label{invac}
 \begin{cases}
 d=d^*(\theta),\\
 q=q^*(\theta)=d^*(\theta)\theta,
 \end{cases}
 \;\;
 \theta>0.
\end{equation}
Let
\begin{equation}\label{la1la}
\tilde\la_1(\theta):=\la_1\left(d_1,q_1,\bm 1-\frac{{\bm v^*}\left(d^*(\theta),q^*(\theta)\right)}{\bm k}\right), \;\;\theta>0.
\end{equation}
Then the semi-trivial equilibrium $(\bm 0,\bm v^*\left(d^*(\theta),q^*(\theta)\right))$ is stable if  $\tilde\la_1(\theta)<0$ and unstable  if $\tilde\la_1(\theta)>0$. 
Noticing that $q_{\bm u}^*(d_1)=q_1$, we have $d^*(q_1/d_1)=d_1$ and 
$$
\tilde\la_1\left({q_1}/{d_1}\right)=\lambda_1\left(d_1, q_1, \bm 1-\frac{{\bm v^*}\left(d_1,q_1\right)}{\bm k}\right)=0.
$$

\begin{theorem}\label{theorem_cobi}
Suppose that $(\bf {H})$ holds, $\bm r\gg\bm 0$, and $d_1,q_1>0$ with $0<q_1<\underline q$. Then for any $\theta>0$, the following statements holds:
\begin{enumerate}
    \item [{$\rm (i)$}] If $ \tilde\la_1(\theta)<0$, then
    \begin{equation*}
      \tilde d^*(\theta) :=\inf\left\{d>0: \ d>d^*(\theta)\ \text{and} \  \la_1\left(d_1,q_1,\bm 1-\frac{{\bm v^*}\left(d,d\theta \right)}{\bm k}\right)\ge 0\right\}
    \end{equation*}
    exists with $d^*(\theta)<\tilde d^*(\theta)$
    such that for any  $(d_2,q_2)$ with $q_2=d_2\theta$ and  $d^*(\theta)<d_2<\tilde d^*(\theta)$
    both semitrivial equilibria $(\bm u^*,\bm 0)$ and $(\bm 0, \bm v^*)$ are locally asymptotically stable and  model \eqref{3p} admits an unstable positive equilibrium.

  \item [{$\rm (ii)$}] If $ \tilde\la_1(\theta)>0$, then
    \begin{equation*}
       \tilde d^*(\theta) :=\sup\left\{d>0: \ d<d^*(\theta)\ \text{and} \  \la_1\left(d_1,q_1,\bm 1-\frac{{\bm v^*}\left(d, d\theta\right)}{\bm k}\right)\le 0\right\}
    \end{equation*}
    exists with $\tilde d^*(\theta)<d^*(\theta)$ such that 
    for any  $(d_2,q_2)$ with $q_2=d_2\theta$ and  $\tilde d^*(\theta)<d_2< d^*(\theta)$
    both semitrivial equilibria $(\bm u^*,\bm 0)$ and $(\bm 0, \bm v^*)$ are unstable and model  \eqref{3p} admits a stable positive equilibrium.
\end{enumerate}

Moreover, $\tilde d^*(\theta)$ satisfies 
     \begin{equation}\label{estima2}
        \begin{cases}
           (\tilde d^*(\theta),\tilde d^*(\theta)\theta)\in G_{12}\;\;\text{if}\;\; \ds\theta>\frac{q_1}{d_1},\\
           (\tilde d^*(\theta),\tilde d^*(\theta)\theta)\in G_{22}\;\;\text{if}\;\; 0<\theta<\ds\f{q_1}{d_1}.\\
        \end{cases}
    \end{equation}
\end{theorem}
\begin{proof}
We prove (i), and (ii) can be proved similarly. Fix $\theta>0$.
Suppose $ \tilde\la_1(\theta)<0$. Let 
$$
A=\left\{d>0: \ d>d^*(\theta)\ \text{and} \  \la_1\left(d_1,q_1,\bm 1-\frac{{\bm v^*}\left(d,d\theta \right)}{\bm k}\right)\ge 0\right\}.
$$
By the proof of Theorem \ref{gdyn} (i), 
$(\bm 0,\bm v^*)$ is unstable or neutrally stable if $(d_2, q_2)\in G_{11}\cup G_{13}$, which yields $A\ne\emptyset$.
Since  $ \tilde\la_1(\theta)<0$,  there exists $\epsilon_0>0$ such that
$$
\la_1\left(d_1,q_1,\bm 1-\frac{{\bm v^*}\left(d^*(\theta)+\epsilon, \theta(d^*(\theta)+\epsilon)\right)}{\bm k}\right)<0, \ \ \text{for any } 0<\epsilon<\epsilon_0.
$$
Therefore, $\tilde d^*(\theta)$ exists with $d^*(\theta)<\tilde d^*(\theta)$.

If $(d_2,q_2)$ satisfies  $q_2=d_2\theta$ and  $\tilde d^*(\theta)<d_2< d^*(\theta)$,  by the definition of $\tilde d^*(\theta)$, we  have  $$\la_1\left(d_1,q_1,\bm{r-v^*}(d_2,q_2)\right)<0,$$ which means that $(\bm 0, \bm v^*)$ is locally asymptotically stable. By Theorem \ref{inva}, $(\bm u^*,\bm 0)$ is also locally asymptotically stable. By the monotone dynamical system theory \cite{hess,hsu1996competitive,smith2008monotone},   model \eqref{3p} admits an unstable positive equilibrium. Finally, it is easy to see that  \eqref{estima2} holds by Theorem \ref{gdyn}.
\end{proof}

\section{Discussions and numerical simulations}

In this section, we discuss about the results of the paper and present some numerical simulations. 

Firstly, we address the impact of spatial heterogeneity on model \eqref{3p}. If the environment is homogeneous, i.e. assumption ({\bf H}) is replaced by $k_1=k_2=k_3$, model \eqref{3p} with $n$ patches has been investigated in our recent paper  \cite{chen2022invasion}. The main results in  \cite{chen2022invasion} can be summarized by Fig. \ref{FigA}. In particular, we prove that the invasion curve is between the lines $q=q_1$ and $q=q_1d/d_1$, $(\bm u^*, \bm 0)$ is globally asymptotically stable in $G_1$, and $(\bm 0, \bm v^*)$ is globally asymptotically stable in $G_2$. These results are independent of the magnitude of drift rate $q_1$ and are similar to the large drift rate case in this paper. Biologically, the downstream end is crowded due to the drift and thereby less friendly compared with the upstream end. If the environment perturbs from being uniformly distributed  and the  upstream locations become advantageous, e.g. assumption ({\bf H}) holds, then a larger drift rate may  compensate for it. This may explain why the homogeneous environment case is similar to the larger drift case in this paper.   

\begin{figure}
 \centering
  \includegraphics[width=0.6\textwidth]{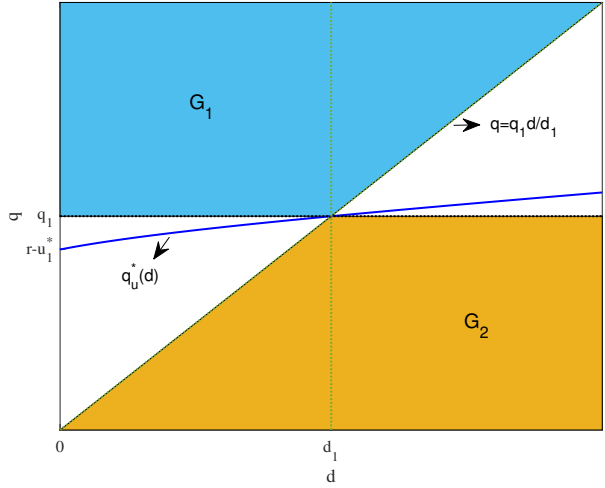}
 \caption{Illustration of the main results for \eqref{3p} with $k_1=k_2=k_3$ in \cite{chen2022invasion}. The blue cure is the invasion curve, which always lies between the lines $q=q_1$ and $q=q_1d/d_1$. Moreover, $(\bm u^*,\bm 0)$ is globally asymptotically stable if $(d_2, q_2)\in G_1$, and  $(\bm 0,\bm v^*)$ is globally asymptotically stable if $(d_2, q_2)\in G_2$.
 }
 \label{FigA}
 \end{figure}

Next, we consider the impact of drift rate on model \eqref{3p}.
By Propositions \ref{profi} and \ref{corollary_global3p1}, if the drift rate $q_1$ is small ($q_1<\underline q$), the invasion curve $q=q^*_{\bm u}(d)$ is defined on a bounded interval and  the species with a smaller random dispersal rate is advantageous; if $q_1$ is large ($q_1>\overline q$), the invasion curve is unbounded  with a slant asymptote $q=\theta d$ for some $\theta>0$ and larger random dispersal rate is favored. The results for the small drift rate case align with the ones in the seminal works \cite{dockery1998evolution,hastings1983can}, which claim that the species with a smaller random dispersal rate will always out-compete the other one in a spatial heterogeneous environment, when both species randomly move in space and are different only by the  movement rate.  When the drift rate becomes large, the outcomes of the competition change dramatically, and the species with a larger dispersal rate may win the competition.   

\begin{figure}
\centering\includegraphics[width=0.6\textwidth]{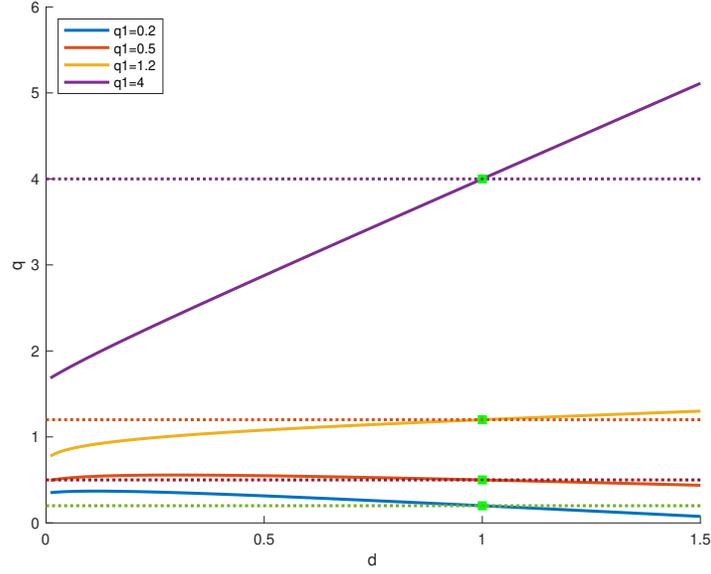}\\
\caption{The invasion curve for different values of $q_1$. Here, $\bm k=(5, 3, 1)$, $\bm r=(1, 2, 1)$, and $d_1=1$. The threshold values for the drift rates are $\underline q=0.4$ and $\overline q=2$.}
\label{interqla1}
\end{figure}

We numerically explore the impact of the drift rate $q_1$ on the shape of the invasion curve $q=q^*_{\bm u}(d)$. Fix $\bm k=(5, 3, 1)$, $\bm r=(1, 2, 1)$, and $d_1=1$. Then we can compute the threshold values for the drift rates: $\underline q=0.4$ and $\overline q=2$. In Fig. \ref{interqla1}, we plot the invasion curves for $q_1=0.2, 0.5, 1.2, 4$. If $q_1=0.2$ or $0.5$, the invasion curves seem to be bounded with $\partial_d q^*_{\bm u}(d_1)<0$, which indicates that a smaller random dispersal rate is favored when $q_1=q_2$ and $d_1\approx d_2$. In sharp contrast with  the smaller drift rate cases, if  $d_1=1.2$ or $4$, the invasion curves seem to be unbounded with $\partial_d q^*_{\bm u}(d_1)>0$. This simulation also shows that  the invasion curve maybe be bounded or unbounded for the intermediate drift case  ($\underline{q}<q_1<\overline{q}$).

Finally, we explore bistablility and coexistence phenomena of  model \eqref{3p}. Let $d_1=1, q_1=1.5$, $\bm r=(3,7,3)$, and $\bm k=(5,3,1)$. We graph the invasion curve $(d^*(\theta), q^*(\theta))$ and $\tilde \lambda(\theta)$ in Fig. \ref{interqla2}.
The stability of $(\bm 0, \bm v^*)$ when $(d_2, q_2)=(d^*(\theta), q^*(\theta))$ is determined by the sign of $\tilde \lambda(\theta)$. In Fig. \ref{interqla2}, we can see that $\tilde \lambda(\theta)$ changes sign, which means that both bistablility and coexistence are possible. Indeed, if we choose $(d_2, q_2)=(3.088, 1.239)$, which is slightly below the invasion curve, then both $(\bm u^*, \bm 0)$ and $(\bm 0, \bm v^*)$ are locally asymptotically stable. As shown in Fig. \ref{interbistable2}, if the initial data is $(\bm u(0),\bm v(0))=((0.1, 0.1, 0.1), (5, 5, 5))$, then the solution of \eqref{3p} converges to $(\bm 0, \bm v^*)$; if the initial data is $(\bm u(0),\bm v(0))=((5, 5, 5), (0.1, 0.1, 0.1))$, then the solution converges to $(\bm u^*, \bm 0)$. Finally, we choose $(d_2, q_2)=(10.28, 0.03)$, which is slightly above the invasion curve ($(\bm u^*, \bm 0)$ is unstable). Since $\tilde\lambda$ is positive, $(\bm 0, \bm v^*)$ is also unstable, and the model has at least one stable positive equilibrium. We graph the solution of \eqref{3p} for initial data $(\bm u(0), \bm v(0))=((1, 1, 1), (1, 1, 1))$ and the solutions seem to converge to a positive equilibrium.

\begin{figure}
\centering\includegraphics[width=1\textwidth]{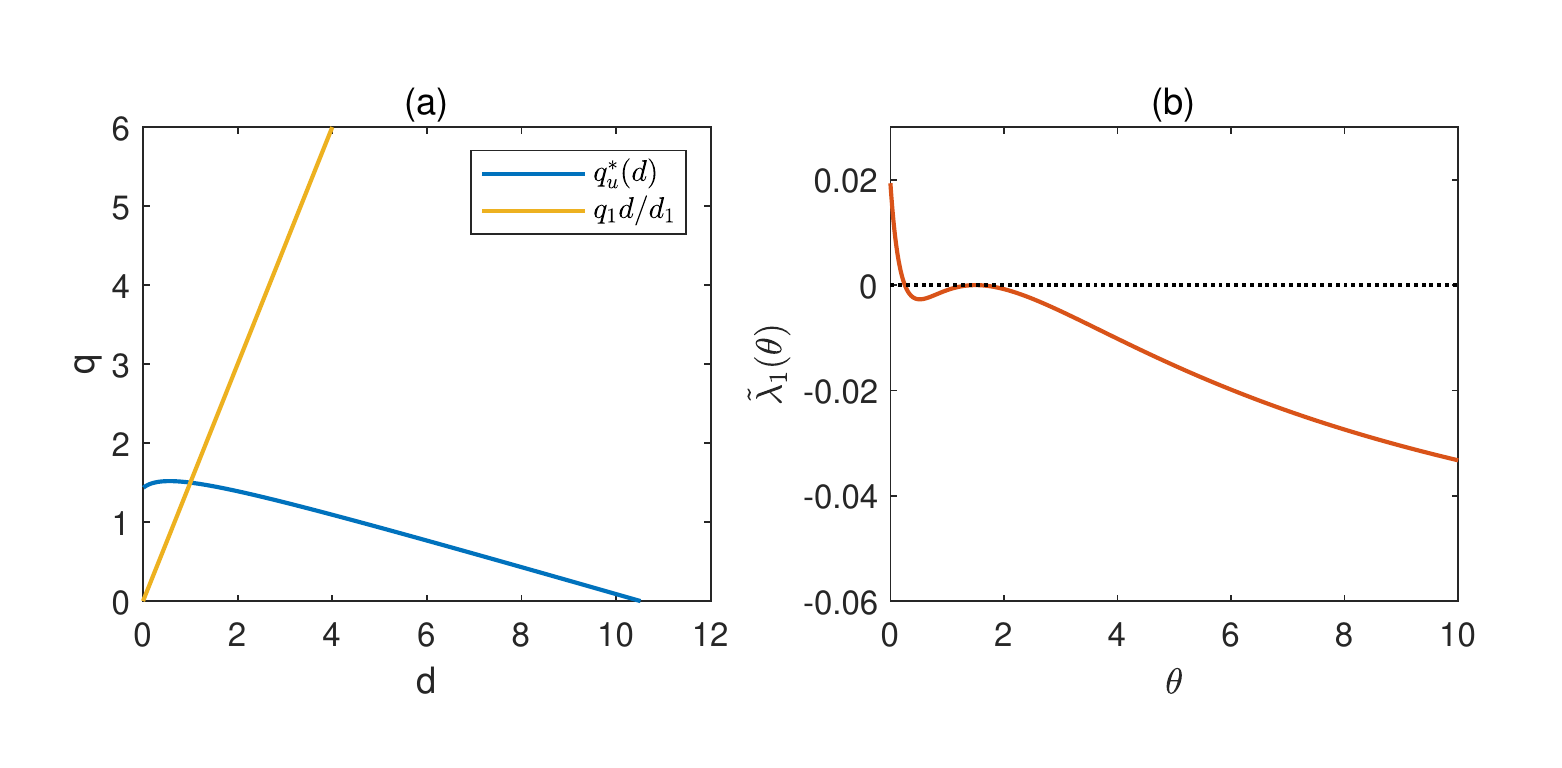}\\
\caption{(a) The invasion curve $(d^*(\theta), q^*(\theta))$ when $d_1=1$, $q_1=1.5$, $\bm r=(3,7,3)$, and $\bm k=(5,3,1)$. (b) The sign of the curve $\tilde{\lambda}_1(\theta)$ determines the stability of $(\bm 0,\bm v^*)$ when $(d_2, q_2)=(d^*(\theta), q^*(\theta))$.}
\label{interqla2}
\end{figure}

\begin{figure}[ht]
\centering\includegraphics[width=1\textwidth]{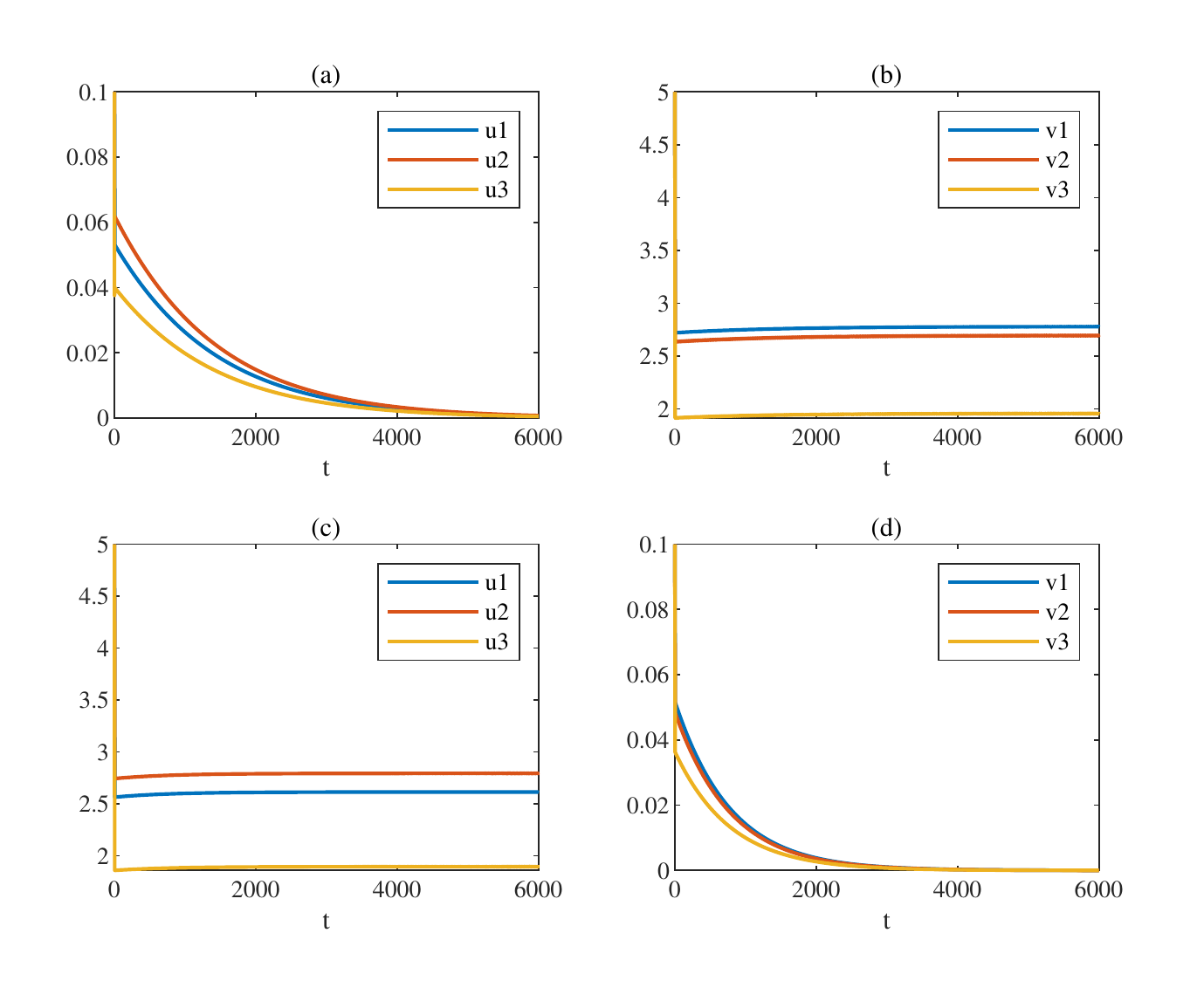}\\
\caption{Solutions of \eqref{3p} with $d_1=1$, $q_1=1.5$, $d_2=3.088$, $q_2=1.239$, $\bm r=(3,7,3)$, $\bm k=(5,3,1)$. For $(a)-(b)$, initial data is $\bm u(0)=(0.1,0.1,0.1)$ and  $\bm v(0)=(5,5,5)$, and species $\bm v$ wins the competition; for $(c)-(d)$, initial data is $\bm u(0)=(5,5,5)$ and $\bm v(0)=(0.1,0.1,0.1)$, and species $\bm u$ wins the competition.}
\label{interbistable2}
\end{figure}

\begin{figure}[ht]
\centering\includegraphics[width=1\textwidth]{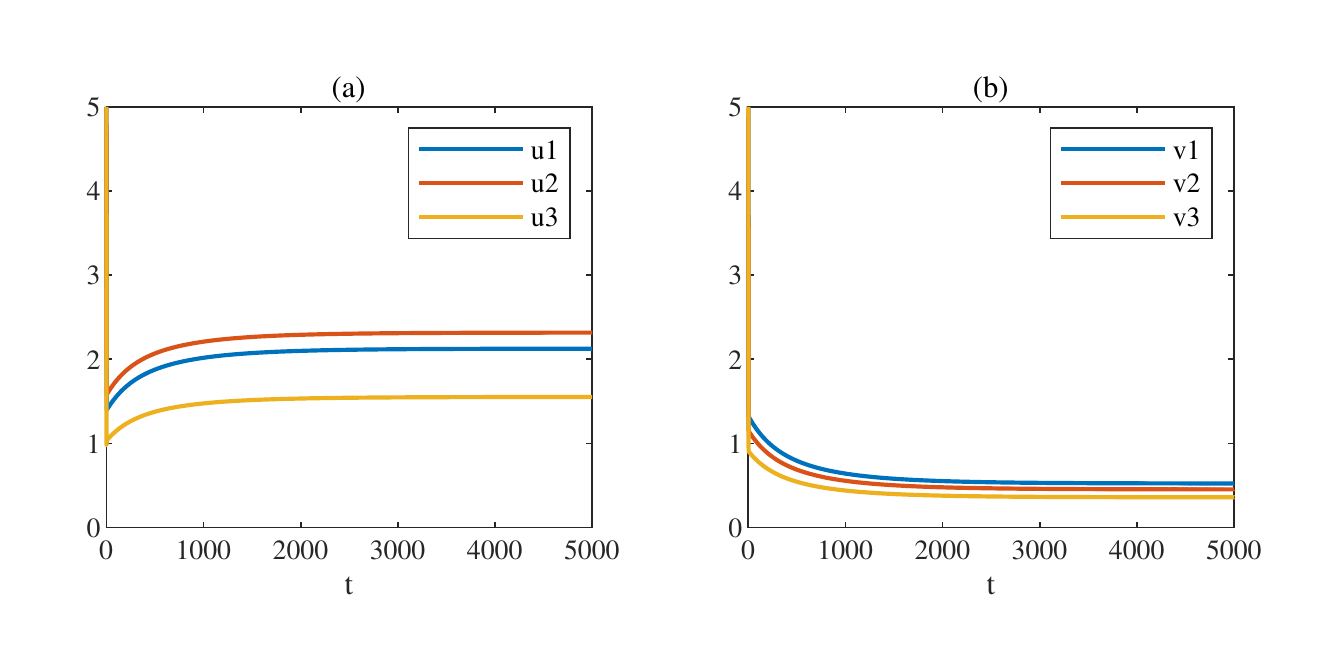}\\

\caption{Solution of \eqref{3p} with $d_1=1$, $q_1=1.5$, $d_2=10.28$, $q_2=0.03$, $\bm r=(3,7,3)$, $\bm k=(5,3,1)$, and initial data is $\bm u(0)=(5,5,5)$ and $\bm v(0)=(5,5,5)$. The two species seem to coexist.}
\label{intercoexist2}
\end{figure}

\section{Proofs for the invasion curve}
In this section, we present the proofs of the results on the invasion curve $q_{\bm u}^*(d)$. We begin with an analysis on $\bm u^*$. A similar result of the following lemma when $r_1=r_2=r_3$ except for the sign of $\sum_{i=1}^3r_i\left(1-{u^*_i}/{k_i}\right)$ can be found in \cite{jiang2020two}.

\begin{lemma}\label{uv}
Suppose that $(\bf {H})$ holds, $\bm r\gg\bm 0$,  $d_1>0$, and $q_1\ge 0$. Then the following statements on $\bm u^*$ hold:
\begin{enumerate}
   \item [${\rm (i)}$] $d_1u_{i+1}^*-(d_1+q_1)u_i^*<0$ for $i=1,2$;
   \item [${\rm (ii)}$] $u^*_1<k_1$ and $u^*_3>k_3$;
 \item [${\rm (iii)}$] If $q_1>\overline q$, then $u^*_1<u^*_2<u^*_3$ and $\sum_{i=1}^3r_i\left(1-\frac{u^*_i}{k_i}\right)>0$;
    \item [${\rm (iv)}$] If $q_1<\underline q$, then $u^*_1>u^*_2>u^*_3$ and  $\sum_{i=1}^3r_i\left(1-\frac{u^*_i}{k_i}\right)<0$.
\end{enumerate}
\end{lemma}
\begin{proof}
By \eqref{3p}, we have
\begin{equation}\label{3TS-1}
\begin{cases}
d_1u^*_2-(d_1+q_1)u^*_1=-r_1u^*_1\left(1-\ds\frac{u^*_1}{k_1}\right),\\
\left(d_1u_3^*-(d_1+q_1)u^*_2\right)-\left(d_1u_2^*-(d_1+q_1)u^*_1\right)=-r_2u^*_2\left(1-\ds\frac{u^*_2}{k_2}\right),\\
d_1u_3^*-(d_1+q_1)u^*_2=r_3u^*_3\left(1-\ds\frac{u^*_3}{k_3}\right).\\
\end{cases}
\end{equation}
Suppose to the contrary that $d_1u^*_2-(d_1+q_1)u^*_1\ge 0$. Then by the first equation of \eqref{3TS-1}, we have  $u_2^*\ge u_1^*\ge k_1$. This, together with assumption $(\bf {H})$ and the second equation of
 \eqref{3TS-1}, implies that $d_1u_3^*-(d_1+q_1)u^*_2>0$ and $u_3^*>k_3$, which contradicts the third equation of \eqref{3TS-1}. Therefore, we have $d_1u^*_2-(d_1+q_1)u^*_1<0$.  Similarly, we can prove $d_1u_3^*-(d_1+q_1)u^*_2<0$. This proves (i).
By (i) and the first and third equations of \eqref{3TS-1}, we can easily obtain (ii).

The proof of (iv) is similar to that of (iii), so we only prove (iii) here.  Suppose $q_1>\overline q$. We rewrite \eqref{3TS-1} as follows:
\begin{equation}\label{3TS}
\begin{cases}
d_1(u^*_2-u^*_1)=-r_1u^*_1\left(1-\ds\frac{q_1}{r_1}-\ds\frac{u^*_1}{k_1}\right),\\
d_1(u^*_3-u^*_2)-(d_1+q_1)(u^*_2-u^*_1)=-r_2u^*_2\left(1-\ds\frac{u^*_2}{k_2}\right),\\
-(d_1+q_1)(u^*_3-u^*_2)=-r_3u^*_3\left(1+\ds\frac{q_1}{r_3}-\ds\frac{u^*_3}{k_3}\right).\\
\end{cases}
\end{equation}
Suppose to the contrary that $u^*_1\ge u^*_2$. Then by the first equation of \eqref{3TS}, we have $k_1-\ds\frac{q_1k_1}{r_1}-u^*_1\ge0$. Since   $q_1>\ds\f{r_1}{k_1}(k_1-k_2)$, we obtain $$k_2>k_1-\ds\frac{q_1k_1}{r_1}\ge u_1^*\ge u_2^*.$$ Then by the second equation of \eqref{3TS}, we get $u^*_2>u^*_3$. This, combined with  $q_1>\ds\f{r_3}{k_3}(k_2-k_3)$, yields $$k_3\left(1+\ds\f{q_1}{r_3}\right)>k_2>u^*_2>u^*_3,$$ which contradicts the last equation of \eqref{3TS}. This proves $u^*_1< u^*_2$.

Suppose to the contrary that $u^*_2\ge u^*_3$. Then by the last equation of \eqref{3TS}, we have $u_3^*\ge k_3\left(1+\ds\f{q_1}{r_3}\right)$. By $q_1>\overline q$, we obtain 
$$
u_2^*\ge u_3^*\ge k_3\left(1+\ds\f{q_1}{r_3}\right)>k_2>k_1-\frac{q_1k_1}{r_1}.
$$
Then by the second equation of \eqref{3TS}, we have $u_1^*>u_2^*$. By the first equation of \eqref{3TS}, we get
$$
0>d_1(u^*_2-u^*_1)=-r_1u^*_1\left(1-\ds\frac{q_1}{r_1}-\ds\frac{u^*_1}{k_1}\right)>0,
$$
which is a contradiction. This proves $u_2^*<u_3^*$. 

Dividing the $i$-th equation of \eqref{3TS} by $u_i^*$, we have
\begin{equation}\label{induct5}
\begin{cases}
-(d_1+q_1)+d_1\frac{u^*_{2}}{u^*_1}+r_1\left(1-\frac{u_1^*}{k_1}\right)=0,\\
(d_1+q_1)\left(\frac{u^*_{1}}{u^*_2}-1\right)-d_1\left(1-\frac{u^*_{3}}{u^*_2}\right)+r_2\left(1-\frac{u_2^*}{k_2}\right)=0,\\
(d_1+q_1) \frac{u^*_{2}}{u^*_3}-d_1 +r_3\left(1-\frac{u_3^*}{k_3}\right)=0.
\end{cases}
\end{equation}
Adding up the equations in \eqref{induct5}, we obtain
\begin{equation}\label{induct6}
\begin{split}
&\sum_{i=2}^{3} \left((d_1+q_1)\left(\frac{u_{i-1}^*}{u_i^* }-1\right) +d_1\left(\frac{u_{i}^*}{u_{i-1}^*}-1\right)\right)+\sum_{i=1}^3 r_i\left(1-\frac{u_i^*}{k_i}\right)\\
=&\sum_{i=2}^{3}\frac{(u_{i-1}^*-u_i^*)((d_1+q_1)u_{i-1}^*-d_1u_i^*)}{u_{i-1}^*u_i^*}+ \sum_{i=1}^3r_i\left(1-\frac{u^*_i}{k_i}\right)=0.
\end{split}
\end{equation}
Then by (i) and $u_1^*<u_2^*<u_3^*$, we have $\sum_{i=1}^3r_i\left(1-\frac{u^*_i}{k_i}\right)>0$. This proves (iii).
\end{proof}

\subsection{Proof of Theorem \ref{inva}}
We prove the existence of the invasion curve $q=q_{\bm u}^*(d)$ in this subsection.

\begin{lemma}\label{eigenv}
Suppose that $(\bf {H})$ holds,  $\bm r\gg\bm 0$, and $d_1,q_1>0$.  Then the following statements hold about the semitrivial equilibrium  $(\bm u^*, \bm 0)$ of \eqref{3p}: 
\begin{itemize}
    \item[{${\rm (i)}$}] If $\sum_{i=1}^3 r_i\left(1-{u_i^*}/{k_i}\right)\ge 0$, then for any $d>0$ there exists $q_{\bm u}^*(d)>0$ such that
$\la_1\left(d,q_{\bm u}^*(d), \bm {1}-{\bm u^*}/{\bm k}\right)=0$, $\la_1\left(d,q, \bm {1}-{\bm u^*}/{\bm k}\right)<0$ for all $q>q_{\bm u}^*(d)$, and $\la_1\left(d,q, \bm {1}-{\bm u^*}/{\bm k}\right)>0$ for all $q<q_{\bm u}^*(d)$;

\item[{${\rm (ii)}$}] If $\sum_{i=1}^3 r_i\left(1-{u_i^*}/{k_i}\right)< 0$, then there exists  $d_0>0$ such that $\la_1(d_0,0,\bm {1}-{\bm u^*}/{\bm k})=0$,  $\la_1(d,0,\bm {1}-{\bm u^*}/{\bm k})<0$ for all $d>d_0$, and $\la_1(d,0,\bm {1}-{\bm u^*}/{\bm k})>0$ for all $d<d_0$. Moreover, the following results hold:
 \begin{enumerate}
\item[{${\rm (ii_1)}$}] For any $d\in(0,d_0)$,  there exists $q_{\bm u}^*(d)>0$ such that the statement in ${\rm (i)}$ holds;
\item[{${\rm (ii_2)}$}] For any $d\in [d_0, \infty)$, we have $\la_1(d,q,\bm {1}-{\bm u^*}/{\bm k})< 0$ for all $q>0$.
\end{enumerate}
\end{itemize}
\end{lemma}
\begin{proof}
For simplicity, we denote $\la_1(d,q):=\la_1\left(d,q,\bm {1}-{\bm u^*}/{\bm k}\right)$. An essential step of the proof is to show the following claim.

\emph{Claim 1:} Fixing $d>0$,  equation  $\lambda_1(d, q)=0$ has at most one root for $q\in [0, \infty)$. 

\noindent\emph{Proof of Claim:} Let $\bm\psi$ be the positive eigenvector corresponding to $\la_1(d,q)$ with  $\sum_{i=1}^3\psi_i=1$.
Then, we have
\begin{equation}\label{esti}
\la_1(d,q)\psi_i=d\sum_{j=1}^3 D_{ij}\psi_j+q\sum_{j=1}^3 Q_{ij}\psi_j+r_i\left(1-\frac{u_i^*}{k_i}\right)\psi_i, \;\;i=1,2,3.
\end{equation}
Differentiating \eqref{esti} with respect to $q$ and denoting $'=\partial/\partial q$, we obtain \begin{equation}\label{esti11}
  \la_1'\psi_i+\la_1\psi'_i=d\sum_{j=1}^3 D_{ij}\psi'_j+\sum_{j=1}^3 Q_{ij}\psi_j+q\sum_{j=1}^3 Q_{ij}\psi'_j+r_i\left(1-\frac{u_i^*}{k_i}\right)\psi'_i, \;\;i=1,2,3.
\end{equation}
Multiplying \eqref{esti} by $\psi_i'$ and \eqref{esti11} by $\psi_i$ and taking the difference of them, we have
\begin{equation}\label{esti2}
  \la_1'\psi^2_i=\sum_{j=1}^3 (dD_{ij}+qQ_{ij})(\psi_i\psi'_j-\psi_i'\psi_j)+\sum_{j=1}^3 Q_{ij}\psi_i\psi_j, \;\;i=1,2,3.
\end{equation}

Motivated by  \cite[Eq. (3.7)]{chen2022invasion}, we introduce $(\beta_1,\beta_2,\beta_3)=\left(1,{d}/(d+q), d^2/(d+q)^2\right)$. Multiplying \eqref{esti2} by $\beta_i$ and summing up in $i$, we obtain
\begin{equation}\label{sp66}
\lambda_1' \sum_{i=1}^3 \beta_i \psi_i^2
=\sum_{i=1}^{2}\left(\frac{d}{d+q}\right)^{i-1}\left(-\psi_i+\frac{d}{d+q}\psi_{i+1}\right)\psi_i.
\end{equation}
Suppose $\la_1\left(d,\tilde q\right)=0$ for some $\tilde q\ge 0$. By Lemma \ref{uv} (ii), we see that
\begin{equation*}
\begin{split}
&d\psi_{2}-(d+\tilde q)\psi_{1}=-\psi_1r_1\left(1-\frac{u^*_1}{k_1}\right)<0,\\
&d\psi_{3}-(d+\tilde q)\psi_{2}=\psi_3r_1\left(1-\frac{u^*_3}{k_3}\right)<0.
\end{split}
\end{equation*}
Therefore by \eqref{sp66}, we have $\lambda'(d, \tilde q)<0$. This proves the claim. 

According to the claim, whether the equation $\lambda_1(d, q)=0$ has a root in $q$ is determined by the sign of $\lambda_1(d, 0)$ and $\lim_{q\to\infty}\la_1(d,q)$. 

\emph{Claim 2:} $\lim_{q\to\infty}\la_1(d,q)<0$.

\noindent \emph{Proof of claim:} Adding up all the equations of \eqref{esti}, we have
\begin{equation}\label{limila}
\la_1(d,q)=\sum_{i=1}^3r_i\left(1-\frac{u_i^*}{k_i}\right)\psi_i,
\end{equation}
which implies that $\la_1(d,q)$ is bounded for $d,q>0$. So up to a subsequence, we may assume $\lim_{q\to\infty}\bm \psi=\tilde{\bm \psi}$.
Dividing \eqref{esti} by $q$ and taking $q\to \infty$, we obtain
\begin{equation*}
\sum_{i=1}^3Q_{ij}\tilde \psi_{j}=0,\;\;i=1,2,3,
\end{equation*}
which yields $\tilde{\bm \psi}=(0,0,1)^T$. Then taking $q\to\infty$ in \eqref{limila}, we have $$\lim_{q\to\infty}\la_1(d,q)=r_3\left(1-\frac{u_3^*}{k_3}\right)<0,$$
where we have used Lemma \ref{uv} (ii) in the last step. This proves the claim.

By Lemma \ref{theorem_quasi}, $\lambda_1(d, 0)$ is strictly decreasing in $d$ with $$
\lim_{d\to0}\la(d,0)=\max_{1\le i\le 3}r_i\left(1-\frac{u_i^*}{k_i}\right)>0\;\; \text{and}\;\; \lim_{d\to\infty}\la(d,0)=\ds\frac{1}{3}\sum_{i=1}^3 r_i\left(1-\frac{u_i^*}{k_i}\right),
$$
where we have used Lemma \ref{uv} (ii) again to see that $1-u_1^*/k_1>0$.
Now the desired results follow from this and Claims 1-2.
\end{proof}

We are ready to prove Theorem \ref{inva}.\\
\emph{Proof of Theorem \ref{inva}.} Let $d_0$ be defined in Lemma \ref{eigenv}, $d^*$ be defined by  \eqref{dstar}, and $q=q^*_{\bm u}(d): (0, d^*)\to \mathbb{R}_+$ be defined in Lemma \ref{eigenv}. Then Theorem \ref{inva} (i)-(ii) follows from Lemmas \ref{uv}-\ref{eigenv} and the fact that the stability/instability of   $(\bm u^*, \bm 0)$ is determined by the sign of $\la_1\left(d,q, \bm {1}-{\bm u^*}/{\bm k}\right)$.
The continuity of $q=q^*_{\bm u}(d)$ follows from $\partial_q \la_1\left(d,q^*_{\bm u}(d), \bm {1}-{\bm u^*}/{\bm k}\right)<0$ (Claim 1 of Lemma \ref{eigenv}) and the implicit function theorem. 
\qed

\subsection{Proof of Propositions \ref{locs}-\ref{profi}}
First we  prove the following useful result:
\begin{lemma}\label{ne0la}
Suppose that $(\bf {H})$ holds, $\bm r\gg\bm 0$,  and $d_1,q_1>0$. Then we have
$$
\la_1\left(d_2,q_2,\bm {1}-\frac{\bm u^*}{\bm k}\right)\ne0,
$$ 
if one of the following conditions holds:
\begin{enumerate}
    \item [{${\rm (i)}$}] $(d_2,q_2)\in G_{11}\cup G_{21}$;
    \item [{${\rm (ii)}$}] $q_1>\overline q$ and $(d_2,q_2)\in G_{12}\cup G_{22}$;
    \item [{${\rm (iii)}$}] $q_1<\underline q$ and  $(d_2,q_2)\in G_{13}\cup G_{23}$.
\end{enumerate}
\end{lemma}
\begin{proof}
Suppose to the contrary that $\lambda_1\left(d_2,q_2,\bm {1}-{\bm u^*}/{\bm k}\right)=0$, and let $\bm \phi\gg0$ be a corresponding eigenvector. Note that $\lambda_1\left(d_1,q_1,\bm {1}-{\bm u^*}/{\bm k}\right)=0$ with a corresponding eigenvector $\bm u^*$.
Let
\begin{equation*}
\tilde  f_0=\tilde f_3=0,\;\;\tilde g_0=\tilde g_3=0,
\end{equation*}
and 
\begin{equation}\label{f}
\tilde f_j=d_1 u^*_{j+1}-(d_1+q_1)u_j^*,\;\;\tilde g_j=d_2\phi_{j+1}-(d_2+q_2)\phi_j, \;\; j=1, 2.
\end{equation}
Then we have
\begin{subequations}\label{sup}
\begin{align}
& \tilde f_{j}-\tilde f_{j-1}=-r_ju^*_j\left(1-\ds\f{u_j^*}{k_j}\right),\;\;\;j=1,2,3,  \label{supp-1}\\
&  \tilde g_{j}-\tilde g_{j-1}=-r_j\phi_j\left(1-\ds\f{u_j^*}{k_j}\right),\;\;\;j=1,2,3.\label{supp-2}
\end{align}
\end{subequations}
 Using similar arguments as in \cite[Lemma 5.7]{chen2022invasion}, we can show
\begin{equation}\label{ineq2mo}
\sum_{j=1}^{2}\left[(d_1-d_2)(\phi_{j+1}-\phi_j)-(q_1-q_2)\phi_j\right]\tilde f_j{\ds\f{d_1^j}{\left(d_1+q_1\right)^{j+1}}}=0,
\end{equation}
and
\begin{equation}\label{ineq3mo}
\sum_{k=1}^{2}\left[(d_2-d_1)(u^*_{j+1}-u^*_j)-(q_2-q_1)u^*_j\right]\tilde g_j{\ds\f{d_2^j}{\left(d_2+q_2\right)^{j+1}}}=0.
\end{equation}
Indeed, multiplying \eqref{supp-1} by $\left(\ds\f{d_1}{d_1+q_1}\right)^j\phi_j$, and summing up from $j=1$ to $j=3$, we have
\begin{equation}\label{sumf}
\begin{split}
&-\sum_{j=1}^3r_ju^*_j\phi_j\left(1-\ds\f{u_j^*}{k_j}\right)\left(\f{d_1}{d_1+q_1}\right)^j\\
=&\sum_{j=1}^{3}\left(\tilde f_j-\tilde f_{j-1}\right)\phi_j\left(\f{d_1}{d_1+q_1}\right)^j\\
=&\tilde f_3\phi_3\left(\f{d_1}{d_1+q_1}\right)^3-\tilde f_{0}\phi_1\left(\f{d_1}{d_1+q_1}\right)+\sum_{j=1}^{2}\tilde f_j\ds\left(\f{d_1}{d_1+q_1}\right)^j\left(\phi_j-\ds\f{d_1}{d_1+q_1}\phi_{j+1}\right)\\
=&-\sum_{j=1}^{2}(d_1\phi_{j+1}-(d_1+q_1)\phi_j)\tilde f_j\ds\f{d_1^j}{\left(d_1+q_1\right)^{j+1}},
\end{split}
\end{equation}
where we have used $\tilde f_3=\tilde f_0=0$ in the last step.
Similarly, multiplying \eqref{supp-2} by $\left(\ds\f{d_1}{d_1+q_1}\right)^ju^*_j$ and summing up from $j=1$ to $j=3$, we obtain
\begin{equation}\label{sumtidg}
\begin{split}
-\sum_{j=1}^3r_ju^*_j\phi_j\left(1-\ds\f{u_j^*}{k_j}\right)\left(\f{d_1}{d_1+q_1}\right)^j
=-\sum_{j=1}^{2} \tilde g_j\tilde f_j\ds\f{d_1^j}{\left(d_1+q_1\right)^{j+1}}.
\end{split}
\end{equation}
Taking the difference of  \eqref{sumf} and \eqref{sumtidg}, we obtain  \eqref{ineq2mo}. Similarly, we can prove \eqref{ineq3mo}.

By Lemma \ref{uv} (i)-(ii) and \eqref{supp-2}, we  have $\tilde f_j, \tilde g_j<0$ for  $j=1,2$. Now we obtain a contradiction for each of (i)-(iii).

(i) We only consider the case  $(d_2,q_2)\in G_{21}$, since the case  $(d_2,q_2)\in G_{11}$ can be studied similarly.  Suppose $(d_2,q_2)\in G_{21}$. Then we have $d_2\le d_1, q_2\le q_1d_2/d_1$, and $(d_1, q_1)\neq (d_2, q_2)$. If $d_1\neq d_2$, then it is easy to check that 
\begin{equation*}
\ds\f{q_1-q_2}{d_1-d_2}\ge\ds\f{q_2}{d_2}.
\end{equation*}
(This inequality can be found in  \cite[Lemma 2.4]{zhou2016lotka}.) This, together with $\tilde g_1,\tilde g_2<0$, yields
\begin{equation}\label{dq2}
(d_1-d_2)(\phi_{j+1}-\phi_j)-(q_1-q_2)\phi_j<0,\;\; j=1,2.
\end{equation}
If $d_1=d_2$, then $q_1>q_2$ and \eqref{dq2} also holds. 
Then by $\tilde f_1,\tilde f_2<0$ and \eqref{ineq2mo}, we have
\begin{equation*}
0<\sum_{j=1}^{2}\left[(d_1-d_2)(\phi_{j+1}-\phi_j)-(q_1-q_2)\phi_j\right]\tilde f_j{\ds\f{d_1^j}{\left(d_1+q_1\right)^{j+1}}}=0,
\end{equation*}
which is a contradiction.


(ii) We only need to obtain a contradiction for the case $(d_2,q_2)\in G_{22}$, since the case $(d_2, q_2)\in G_{12}$ can be studied similarly.
Suppose $(d_2,q_2)\in G_{22}$. Then we have $d_2>d_1$ and $q_2\le q_1$. By Lemma \ref{uv} (i), we have
$u_1^*<u_2^*<u^*_3$, which implies that
$$
(d_2-d_1)(u^*_{j+1}-u^*_j)-(q_2-q_1)u^*_j>0, \ \ j=1, 2.
$$
This, combined with  $\tilde g_1, \tilde g_2<0$ and \eqref{ineq3mo}, gives a contradiction.

(iii) We only obtain a contradiction for the case $(d_2,q_2)\in G_{23}$, since the case $(d_2,q_2)\in G_{13}$ can be studied similarly.
Suppose $(d_2,q_2)\in G_{23}$. Then we have $d_2\le d_1$, $q_2\le q_1$, and $(d_1, q_1)\neq (d_2, q_2)$. By Lemma \ref{uv} (ii), we have
$u_1^*>u_2^*>u^*_3$, which implies that
$$
(d_2-d_1)(u^*_{j+1}-u^*_j)-(q_2-q_1)u^*_j>0, \ \ j=1, 2.
$$
This combined with  $\tilde g_1, \tilde g_2<0$ and \eqref{ineq3mo} gives a contradiction.
\end{proof}

We are ready to prove Propositions \ref{locs}-\ref{profi}. 

\noindent\emph{Proof of Proposition \ref{locs}.} (i) We only prove the case $G_{21}\subset S_2$, (i.e., $(\bm u^*,\bm 0)$ is unstable for $(d_2,q_2)\in G_{21}$), since the case $ G_{11}\subset S_1$ can be proved similarly.
To avoid confusion, we denote $\bm u^*$ by $\bm u^*_{\bm k}$.  It is easy to see that $\bm u^*_{\bm k}$ depends continuously on $\bm k$.

Suppose $(d_2, q_2)\in G_{21}$. We need to prove that $\lambda_1\left(d_2,q_2,\bm {1}-{\bm u^*_{\bm k}}/{\bm k}\right)>0$. Suppose to the contrary that  $\lambda_1\left(d_2,q_2,\bm {1}-{\bm u^*_{\bm k}}/{\bm k}\right)\le 0$. By Lemma \ref{ne0la}, we must have $\lambda_1\left(d_2,q_2,\bm {1}-{\bm u^*_{\bm k}}/{\bm k}\right)< 0$.
By \cite[Theorem 4.2]{chen2022invasion}, we have $\lambda_1\left(d_2,q_2,\bm {1}-{\bm u^*_{\bm k'}}/{\bm k'}\right)>0$, where $\bm k'=(k_3,k_3,k_3)$.

Let $\Lambda(s):=\lambda_1\left(d_2,q_2,\bm {1}-{\bm u^*_{\bm k_1(s)}}/{\bm k_1(s)}\right)$, where
$\bm k_1(s)=s\bm k+(1-s)\bm k'$ satisfies ({\bf H}) for any $s\in [0, 1]$. Since
$$
\Lambda(1)=\lambda_1\left(d_2,q_2,\bm {1}-\frac{\bm u^*_{\bm k}}{\bm k}\right)<0\;\;\text{ and}\;\;
\Lambda(0)=\lambda_1\left(d_2,q_2,\bm {1}-\frac{\bm u^*_{\bm k'}}{\bm k'}\right)>0,$$
there exists $s_0\in(0,1)$ such that
$\Lambda(s_0)=0$, which contradicts Lemma \ref{ne0la}.

(ii) By \cite[Theorem 4.2]{chen2022invasion}, if $(d_2,q_2)\in G_{22}$ then $\lambda_1\left(d_2,q_2,\bm {1}-{\bm u^*_{\bm k'}}/{\bm k'}\right)>0$; and if $(d_2,q_2)\in G_{12}$, then $\lambda_1\left(d_2,q_2,\bm {1}-{\bm u^*_{\bm k'}}/{\bm k'}\right)<0$,
where $\bm k'=(k_3,k_3,k_3)$. Then using similar arguments as (i), we can prove (ii).

(iii) Let $\bm\psi$ be the positive eigenvector corresponding to $\lambda_1:=\la_1\left(d_2,q_1,\bm {1}-{\bm u^*}/{\bm k}\right)$ with  $\sum_{i=1}^3\psi_i=1$.
Then, we have
\begin{equation}\label{esti00}
\la_1\psi_i=d_2\sum_{j=1}^3 D_{ij}\psi_j+q_1\sum_{j=1}^3Q_{ij}\psi_j+r_i\left(1-\frac{u_i^*}{k_i}\right)\psi_i, \;\;i=1,2,3.
\end{equation}
Differentiating \eqref{esti} with respect to $d_2$ and denoting $'=\partial/\partial d_2$, we obtain \begin{equation}\label{esti100}
  \la_1'\psi_i+\la_1\psi'_i=d_2\sum_{j=1}^3 D_{ij}\psi'_j+\sum_{j=1}^3 D_{ij}\psi_j+q_1\sum_{j=1}^3Q_{ij}\psi'_j+r_i\left(1-\frac{u_i^*}{k_i}\right)\psi'_i, \;\;i=1,2,3.
\end{equation}
Multiplying \eqref{esti00} by $\psi_i'$ and \eqref{esti100} by $\psi_i$ and taking the difference of them, we have
\begin{equation}\label{esti200}
  \la_1'\psi^2_i=\sum_{j=1}^3 (d_2D_{ij}+q_1Q_{ij})(\psi_i\psi'_j-\psi_i'\psi_j)+\sum_{j=1}^3 D_{ij}\psi_i\psi_j, \;\;i=1,2,3.
\end{equation}

Similar to the proof of Lemma \ref{eigenv}, let $(\beta_1,\beta_2,\beta_3)=\left(1,{d_2}/(d_2+q_1), d_2^2/(d_2+q_1)^2\right)$. Multiplying \eqref{esti200} by $\beta_i$ and adding up them in $i$, we obtain
\begin{equation}\label{sp6600}
\lambda_1' \sum_{i=1}^3 \beta_i \psi_i^2
=\sum_{i=1}^{2}\left(\frac{d_2}{d_2+q_1}\right)^{i-1}\left(-\psi_i+\frac{d_2}{d_2+q_1}\psi_{i+1}\right)(\psi_i-\psi_{i+1}).
\end{equation} 
Note that $\la_1\left(d_1,q_1,\bm {1}-{\bm u^*}/{\bm k}\right)=0$ with a corresponding eigenvector $(u_1^*, u^*_2, u_3^*)^T$. Moreover, by Lemma \ref{uv}, we have $d_1u_{i+1}^*-(d_1+q_1)u_i^*<0$ for $i=1, 2$ and $u_1^*>u^*_2>u_3^*$. Then it follows from \eqref{sp6600} that 
\begin{equation}\label{smaqsign}
\left.\f{\partial \la_1\left(d_2,q_1,\bm {1}-{\bm u^*}/{\bm k}\right)}{\partial d_2}\right|_{d_2=d_1}<0.
\end{equation}
This implies that $\la_1\left(d_2,q_1,\bm {1}-{\bm u^*}/{\bm k}\right)>0$ if $0<d_1-d_2\ll1$ and $\la_1\left(d_2,q_1,\bm {1}-\frac{\bm u^*}{\bm k}\right)<0$ if
$0<d_2-d_1\ll 1$. Then, by Lemma \ref{ne0la}, we have $G_{13}\subset S_1$ and $G_{23}\subset S_2$.
\qed

\begin{remark}
A similar inequality of  \eqref{smaqsign} is proved in \cite{jiang2020two}, and we include the proof for completeness here.
\end{remark}


\noindent\emph{Proof of Proposition \ref{profi}.} For any $0<d<d^*$, let $\bm\psi$ be the eigenvector corresponding to $\la_1\left(d,q_{\bm u}^*(d),\bm {1}-{\bm u^*}/{\bm k}\right)=0$ with $\bm\psi\gg0$ and $\sum_{i=1}^3\psi_i=1$. Then
\begin{equation}\label{esti1}
d\sum_{j=1}^3 D_{ij}\psi_j+q_{\bm u}^*(d)\sum_{j=1}^3 Q_{ij}\psi_j+r_i\left(1-\frac{u_i^*}{k_i}\right)\psi_i=0, \;\;i=1,2,3.
\end{equation}

(i) Up to a subsequence, we may assume $\lim_{d\to0}\bm \psi=\tilde{\bm \psi}$ for some $\tilde{\bm \psi}\ge \bm 0$ and  $\sum_{i=1}^3\tilde\psi_i=1$.
We first claim that  $q_{\bm u}^*(d)$ is bounded for $d\in (0,\delta)$ with $\delta\ll1$.
If it is not true, then dividing \eqref{esti1} by $q_{\bm u}^*(d)$ and taking $d\to0$, we have
\begin{equation}\label{esti4}
\sum_{j=1}^3Q_{ij}\tilde \psi_{j}=0,\;\;i=1,2,3,
\end{equation}
which yields $\tilde{\bm \psi}=(0,0,1)^T$.
Adding up all the equations of \eqref{esti1}, we have
\begin{equation}\label{esti3}
\sum_{i=1}^3r_i\left(1-\frac{u_i^*}{k_i}\right)\tilde\psi_i=0.
\end{equation}
Taking $d\to0$ in \eqref{esti3}, we have
$k_3-u_3^*=0$, which contradicts Lemma \ref{uv} (ii). This proves the claim.
By the claim, up to a subsequence, we may assume $\lim_{d\to0}q^*(\theta)=\tilde q_0\in[0,\infty)$.
Consequently,  for sufficiently small $\epsilon>0$, there exists $\bar d>0$ such that $q_{\bm u}^*(d)<\tilde q_0+\epsilon$ for all $0<d<\bar d$.
It follows from Lemma \ref{eigenv} that
\begin{equation}\label{qbound}
\la_1\left(d, \tilde q_0+\epsilon,  \bm {1}-\frac{\bm u^*}{\bm k}\right)<0
\end{equation}
for all  $0<d<\bar d$.
Hence,
\begin{equation*}
\begin{split}
&\lim_{d\to 0}\la_1\left(d, \tilde q_0+\epsilon,  \bm {1}-\frac{\bm u^*}{\bm k}\right)\\=&\ds\max\left\{r_1\left(1-\frac{u_1^*}{k_1}\right)-(\tilde q_0+\epsilon),r_2\left(1-\frac{u_2^*}{k_2}\right)-(\tilde q_0+\epsilon),r_3\left(1-\frac{u_3^*}{k_3}\right)\right\}\le 0.\\
\end{split}
\end{equation*}
Since $k_3-u_3^*=0$ and $\epsilon>0$ was arbitrary, 
\begin{equation}\label{supp}
\ds\max\left\{r_1\left(1-\frac{u_1^*}{k_1}\right)-\tilde q_0,r_2\left(1-\frac{u_2^*}{k_2}\right)-\tilde q_0\right\}\le0.
\end{equation}
Therefore, we have $\tilde q_0\ge q_0>0$.
Similarly, we can prove  $\tilde q_0\le q_0$. This proves (i).

Now we prove (ii)-(iv). If we show that
\eqref{qinf1} holds when $\sum_{i=1}^3 r_i\left(1-{u_i^*}/{k_i}\right)< 0$, and 
\eqref{qinf2} holds when $\sum_{i=1}^3 r_i\left(1-{u_i^*}/{k_i}\right)\ge 0$.
Then (iv) holds, and (ii)-(iii) follow from Theorem \ref{uv} (iii)-(iv).

By Lemma \ref{eigenv} (i), the function $q=q_{\bm u}^*(d)$ is defined for $d\in(0,\infty)$ when $\sum_{i=1}^3 r_i\left(1-{u_i^*}/{k_i}\right)\ge 0$.
We claim that ${q_{\bm u}^*(d)}/{d}$ is bounded for $d\in(\delta,\infty)$ for any fixed $\delta>1$. If it is not true, up to a subsequence, we may assume
\begin{equation*}
\lim_{d\to \infty}\ds\f{q_{\bm u}^*(d)}{d}=\infty\;\;
\text{and}\;\;\lim_{d\to \infty} \bm\psi=\hat{\bm\psi}
\end{equation*}
for some $\hat{\bm \psi}\ge \bm 0$ and  $\sum_{i=1}^3\hat\psi_i=1$. Then dividing \eqref{esti1} by $q_{\bm u}^*(d)$ and taking $d\to\infty$, we can obtain a contradiction  using similar arguments as in the proof of (i). Therefore, ${q_{\bm u}^*(d)}/{d}$ is bounded for $d\in(\delta,\infty)$. Then using similar arguments as in the proof of \cite[Proposition 4.4]{chen2022invasion}, we can show that \eqref{qinf2} holds.

 By Lemma \ref{eigenv} (ii), the function $q=q_{\bm u}^*(d)$ is defined for $d\in(0,d_0)$   
 when \break
 $\sum_{i=1}^3 r_i\left(1-{u_i^*}/{k_i}\right)< 0$. Using similar arguments as in (i), we can show that $q_{\bm u}^*(d)$ is bounded for $d\in(d_0-\delta,d_0)$ for some $\delta\ll1$. Then,
up to a subsequence, we may assume
\begin{equation*}
\lim_{d\to d_0} q_{\bm u}^*(d)=\eta\;\;\text{and}\;\;\lim_{d\to d_0} \bm\psi=\bm\psi^*
\end{equation*}
for some ${\bm \psi^*}\ge \bm 0$ and  $\sum_{i=1}^3\psi^*_i=1$. Taking $d\to d_0$ in \eqref{esti1}, we see that
\begin{equation}\label{qdinfty0}
d_0\sum_{j=1}^3 D_{ij}\psi^*_j+\eta\sum_{j=1}^3 Q_{ij}\psi^*_j+r_i\left(1-\frac{u_i^*}{k_i}\right)\psi^*_i=0, \;\;i=1,2,3,
\end{equation}
which yields $\la_1\left(d_0,\eta, \bm {1}-{\bm u^*}/{\bm k}\right)=0$. By the proof of Lemma \ref{eigenv},  $\la_1\left(d_0,q, \bm {1}-{\bm u^*}/{\bm k}\right)=0$ has at most one root for $q\in [0,\infty)$. Since $\la_1\left(d_0,0, \bm {1}-{\bm u^*}/{\bm k}\right)=0$, we must have  $\eta=0$. This proves \eqref{qinf1}.\qed

\section{Proofs for the competitive exclusion results}

Let $(\bm u,\bm v)$ be a positive equilibrium of model \eqref{3p}. Define
\begin{equation}\label{figi0}
\begin{split}
&f_0=f_3=0,\;\;g_0=g_3=0,\\
&f_j=d_1u_{j+1}-(d_1+q_1)u_j,\;\;g_j=d_2v_{j+1}-(d_2+q_2)v_j,\;\; j=1,2.
\end{split}
\end{equation}
Clearly, we have 
\begin{equation}\label{f-k0}
f_{j}-f_{j-1}=-r_ju_j\left(1-\ds\f{u_j+v_j}{k_j}\right),\;\;\;j=1,2,3,
\end{equation}
and
\begin{equation}\label{g-k0}
g_{j}-g_{j-1}=-r_jv_j\left(1-\ds\f{u_j+v_j}{k_j}\right),\;\;\;j=1,2,3.
\end{equation}

Then we have the following result about  the sign of $f_j,g_j$, $j=1,2$.
\begin{lemma}\label{sign}
Suppose that $(\bf {H})$ holds, $\bm r\gg \bm 0$, and $d_1, q_1, d_2, q_2>0$. If $(\bm u,\bm v)$ is a positive equilibrium of model \eqref{3p}, then we have $f_1,g_1,f_2,g_2<0$.
\end{lemma}
\begin{proof}
First we prove $f_1<0$. Suppose to the contrary that $f_1\ge 0$. By \eqref{f-k0}-\eqref{g-k0}, we have $k_1-u_1-v_1\le0$ and $g_1\ge0$. Since $f_1,g_1\ge 0$, we have $u_2>u_1$ and $v_2>v_1$. This combined with $(\bf {H})$ implies that $k_2-u_2-v_2<0$. Then by \eqref{f-k0}-\eqref{g-k0} again, we obtain that $f_2,g_2>0$ and $k_3-u_3-v_3<0$, which contradicts \eqref{f-k0} with $j=3$. Therefore, we have $f_1<0$. Consequently, by  \eqref{f-k0}-\eqref{g-k0} with $j=1$, we have $g_1<0$. Using similar arguments, we can prove $f_2,g_2<0$.
\end{proof}

The following result is similar to \cite[Lemma 5.7]{chen2022invasion} with $j=1$ and $j=n=3$ (see also the proof of Lemma \ref{ne0la}). Thus, we omit the proof.
\begin{lemma}\label{indfg2}
Suppose that $(\bf {H})$ holds, $\bm r\gg \bm 0$, and $d_1, q_1, d_2, q_2>0$. If $(\bm u,\bm v)$ is a positive equilibrium of model \eqref{3p},  then the following equations hold:
\begin{equation}\label{fg0}
\sum_{j=1}^{2}\left[(d_1-d_2)(v_{j+1}-v_j)-(q_1-q_2)v_j\right]f_j{\ds\f{d_1^j}{\left(d_1+q_1\right)^{j+1}}}=0,
\end{equation}
and
\begin{equation}\label{gf20}
\sum_{j=1}^{2}\left[(d_2-d_1)(u_{j+1}-u_j)-(q_2-q_1)u_j\right]g_k{\ds\f{d_2^j}{\left(d_2+q_2\right)^{j+1}}}=0.
\end{equation}
\end{lemma}

An essential step to prove the competitive exclusion results for model \eqref{3p} is to show the nonexistence of positive equilibrium: 
\begin{lemma}\label{nonex}
Suppose that $(\bf {H})$ holds, $\bm r\gg \bm 0$, and $d_1,q_1>0$. Let $G^*_{12}$ and $G_{23}^*$ be defined by \eqref{G1213s} and \eqref{G23*0}, respectively.
Then model \eqref{3p} admits no positive equilibrium, if one of the following conditions holds:
\begin{enumerate}
    \item [{${\rm (i)}$}] $(d_2,q_2)\in G_{11}\cup G_{21}$;
     \item [{${\rm (ii)}$}] $q_1>\overline q$ and $(d_2,q_2)\in G_{12}\cup G_{22}\cup G_{23}^*$;
       \item [{${\rm (iii)}$}] $q_1<\underline q$ and  $(d_2,q_2)\in G_{13}\cup G_{23}\cup G^*_{12}$.
    \item [{${\rm (iv)}$}] $\underline q\le q_1\le\overline q$ and $(d_2,q_2)\in G^*_{12}\cup G_{23}^*$;
   
\end{enumerate}
\end{lemma}
\begin{proof}
Suppose to the contrary that model \eqref{3p} admits a positive equilibrium $(\bm u,\bm v)$. Then we will obtain a contradiction for each of the cases (i)-(iv).

(i) We only consider the case $(d_2,q_2)\in  G_{21}$. Since the nonlinear terms of \eqref{3p} are symmetric, the case  $(d_2,q_2)\in  G_{11}$ can be proved similarly. Suppose $(d_2,q_2)\in  G_{21}$. Then we have $d_2\le d_1, q_2\le q_1d_2/d_1$ and $(d_1, q_1)\neq (d_2, q_2)$. 
First, we claim that
\begin{equation}\label{dq}
(d_1-d_2)(v_{j+1}-v_j)-(q_1-q_2)v_j<0\;\;\text{for}\;\; j=1,2.
\end{equation}
Indeed if $d_1=d_2$, then $q_1>q_2$ and \eqref{dq} holds.
If $d_1>d_2$, then it is easy to check that
\begin{equation*}
\ds\f{q_1-q_2}{d_1-d_2}\ge\ds\f{q_2}{d_2}.
\end{equation*}
(This inequality is in \cite[Lemma 2.4]{zhou2016lotka}). This, combined with $g_1,g_2<0$, proves  \eqref{dq}.
Then by $f_1, f_2<0$ and  \eqref{fg0}, we have
\begin{equation*}
0<\sum_{j=1}^{2}\left[(d_1-d_2)(v_{j+1}-v_j)-(q_1-q_2)v_j\right]f_j{\ds\f{d_1^j}{\left(d_1+q_1\right)^{j+1}}}=0,
\end{equation*}
which is a contradiction.

(ii) We first consider the case  $(d_2,q_2)\in  G_{22}$. Since the nonlinear terms of \eqref{3p} are symmetric, the case  $(d_2,q_2)\in  G_{12}$ can be proved similarly. Suppose $(d_2, q_2)\in G_{22}$. Then, $d_2>d_1$ and $q_2\le q_1$.
By \eqref{3p}, we have
\begin{subequations}\label{3sequv}
\begin{align}
&d_1(u_2-u_1)=-r_1u_1\left(1-\ds\f{q_1}{r_1}-\ds\f{u_1+v_1}{k_1}\right),\label{3sequv1}\\
&d_2(v_2-v_1)=-r_1v_1\left(1-\ds\f{q_2}{r_2}-\ds\f{u_1+v_1}{k_1}\right),\label{3sequv2}\\
&(d_1+q_1)(u_1-u_2)-d_1(u_2-u_3)=-r_2u_2\left(1-\ds\f{u_2+v_2}{k_2}\right),\label{3sequv3}\\
&(d_2+q_2)(v_1-v_2)-d_2(v_2-v_3)=-r_2v_2\left(1-\ds\f{u_2+v_2}{k_2}\right),\label{3sequv4}\\
&(d_1+q_1)(u_2-u_3)=-r_3u_3\left(1+\ds\f{q_1}{r_3}-\ds\f{u_3+v_3}{k_3}\right),\label{3sequv5}\\
&(d_2+q_2)(v_2-v_3)=-r_3u_3\left(1+\ds\f{q_2}{r_3}-\ds\f{u_3+v_3}{k_3}\right).\label{3sequv6}
\end{align}
\end{subequations}
Then we show that $u_1<u_2<u_3$.
Suppose to the contrary that $u_1\ge u_2$. Then, by \eqref{3sequv1}, we see that $$u_1+v_1\le k_1-\ds\f{q_1k_1}{r_1}\le  k_1-\ds\f{q_2k_1}{r_1},$$
where we have used $q_2\le q_1$ in the last inequality.
This, combined with \eqref{3sequv2}, implies that
$v_1\ge v_2$. Noticing that $$q_1>\overline q\ge \ds\f{r_1}{k_1}(k_1-k_2),$$ we have 
$$
u_2+v_2\le u_1+v_1\le k_1-\ds\f{q_1k_1}{r_1} <k_2,
$$
and consequently $u_3<u_2$ and $v_3<v_2$ by \eqref{3sequv3}-\eqref{3sequv4}. This, combined with $q_1>\overline q$, implies that 
$$
u_3+v_3<u_2+v_2<k_2<k_3+\ds\f{k_3q_1}{r_3},
$$
which contradicts \eqref{3sequv5}.
Similarly, we can show that $u_2<u_3$.
So, $u_1<u_2<u_3$, which leads to
$$
(d_2-d_1)(u_{j+1}-u_j)-(q_2-q_1)u_j>0, \ \ j=1, 2.
$$
Then by Lemma \ref{sign} and  \eqref{gf20}, we have
\begin{equation*}
0>\sum_{j=1}^{2}\left[(d_2-d_1)(u_{j+1}-u_j)-(q_2-q_1)u_j\right]g_j{\ds\f{d_2^j}{\left(d_2+q_2\right)^{j+1}}}=0,
\end{equation*}
which is a contradiction.

Now suppose that $(d_2,q_2)\in G_{23}^*$. Then $q_2<\underline q$ and $(d_1,q_1)\in \hat{G}_{13}$, where
\begin{equation}\label{hatg13}
\hat G_{13}:=\{(d, q):\; d\ge d_2,q_2\le q< \frac{q_2}{d_2}d,(d,q)\ne(d_2,q_2)\}.
\end{equation}
Since the nonlinear terms of \eqref{3p} are symmetric, this case can be proved similarly as   the case $(d_2,q_2)\in G_{13}$ (the proof is immediately below).

(iii) Suppose that $(d_2,q_2)\in G_{13}$. Then we have $d_2\ge d_1$, $q_1\le q_2\le q_1d_2/d_1$, and $(d_1, q_1)\neq (d_2, q_2)$. We show that $u_1>u_2>u_3$. Suppose to the contrary that $u_1\le u_2$. Then, by \eqref{3sequv1} and $q_2\ge q_1$, we have
\begin{equation*}
u_1+v_1\ge k_1-\ds\f{q_1k_1}{r_1}>k_1-\ds\f{q_2k_1}{r_1}.
\end{equation*}
This, combined with \eqref{3sequv2}, yields $v_1\le v_2$. Noting that
\begin{equation*}
q_1<\underline{q}<\f{r_1}{k_1}(k_1-k_2),
\end{equation*}
we have
\begin{equation*}
k_2<k_1-\f{q_1k_1}{r_1}\le u_1+v_1\le u_2+v_2.
\end{equation*}
Then by \eqref{3sequv3} and \eqref{3sequv4}, we have $u_3>u_2$ and $v_3>v_2$. Since $q_1<\underline{q}$, we have
\begin{equation*}
k_3+\f{k_3q_1}{r_3}<k_2<u_2+v_2<u_3+v_3,
\end{equation*}
which contradicts \eqref{3sequv5}. Similarly, we can show  $u_2>u_3$. Therefore, we have $u_1>u_2>u_3$ and
\begin{equation*}
(d_2-d_1)(u_{j+1}-u_{j})-(q_2-q_1)u_j<0,\;\;j=1,2.
\end{equation*}
Then it follows from Lemma \ref{sign} and  \eqref{gf20} that
\begin{equation*}
0<\sum_{j=1}^{2}\left[(d_2-d_1)(u_{j+1}-u_j)-(q_2-q_1)u_j\right]g_j{\ds\f{d_2^j}{\left(d_2+q_2\right)^{j+1}}}=0,
\end{equation*}
which is a contradiction.

For the case $(d_2,q_2)\in  G_{23}$, using similar arguments as above, we can obtain  $v_1>v_2>v_3$,
which leads to
$$
(d_1-d_2)(v_{j+1}-v_j)-(q_1-q_2)v_j<0, \ \ j=1, 2.
$$
This, combined with Lemma \ref{sign} and  \eqref{fg0}, implies that
\begin{equation*}
0<\sum_{j=1}^{2}\left[(d_1-d_2)(v_{j+1}-v_j)-(q_1-q_2)v_j\right]f_j{\ds\f{d_1^j}{\left(d_1+q_1\right)^{j+1}}}=0,
\end{equation*}
which is a contradiction.

Next suppose that $(d_2,q_2)\in G_{12}^*$. Then $q_2>\overline q$ and $(d_1,q_1)\in \hat{G}_{22}$, where
\begin{equation*}
\hat G_{22}:=\{(d,q):\; d>d_2,0<q\le q_2\}.
\end{equation*}
Since the nonlinear terms of \eqref{3p} are symmetric, this case can be proved similarly the case $(d_2,q_2)\in G_{22}$ in (ii).

(iv) If $(d_2,q_2)\in G_{23}^*$, the proof is the similar the corresponding case in (ii); If $(d_2,q_2)\in G_{12}^*$, the proof is the similar the corresponding case in (iii).
\end{proof}

We are ready to prove Theorems \ref{gdyn}, \ref{gdyn1} and \ref{gdyn2}.

\noindent\emph{Proof of Theorem \ref{gdyn}.} 
(i) Suppose that $q_1<\underline q$ and $(d_2,q_2)\in G_{21}\cup G_{23}$. By Lemma \ref{nonex} (i) and (iii),
model \eqref{3p} admits no positive equilibrium.
By Theorem \ref{locs} (i) and (iii), $(\bm u^*, \bm 0)$ is unstable. Then it follows from  the monotone dynamical system theory \cite{hess,hsu1996competitive, LAM2016Munther,smith2008monotone} that  $(\bm u^*,\bm 0)$ is globally asymptotically stable.


(ii) Suppose that $q_1<\underline q$ and $(d_2,q_2)\in  G_{11}\cup G_{12}^*$.  By Lemma \ref{nonex} (i) and (iii),
model \eqref{3p} admits no positive equilibrium.  By the monotone dynamical system theory \cite{hess,hsu1996competitive, LAM2016Munther,smith2008monotone}, it suffices to show that
$(\bm 0,\bm v^*)$ is unstable. If $(d_1,q_1)\in G_{12}^*$, then $q_2>\overline q$ and $(d_1,q_1)\in \tilde G_{22}:=\{(d,q):d>d_2,0<q\le q_2\}$. Since the nonlinear terms of model \eqref{3p} are
symmetric, it follows from Proposition \ref{locs} (ii) that $(\bm 0,\bm v^*)$ is unstable. If  $(d_2,q_2)\in G_{11}$, then $(d_1,q_1)\in \tilde G_{21}$, where
$$
\tilde G_{21}:=\left\{(d,q):0<d \le d_2,0<q\le \ds\f{q_2}{d_2}d,(d,q)\ne(d_2,q_2)\right\}.
$$
Similarly, it follows from Proposition \ref{locs} (i) that $(\bm 0,\bm v^*)$ is unstable.

Finally,  suppose that $q_1<\underline q$ and $(d_2,q_2)\in G_{13}$. By Proposition \ref{locs} (iii), $(\bm u^*, \bm 0)$ is locally asymptotically stable. By Lemma \ref{nonex}  (iv),
model \eqref{3p} admits no positive equilibrium. If  $(\bm 0, \bm v^*)$ is   locally asymptotically stable, then model \eqref{3p} admits one unstable positive steady state, which is a contradiction. If  $(\bm 0, \bm v^*)$ is unstable, then the  monotone dynamical system theory \cite{hess,hsu1996competitive, LAM2016Munther,smith2008monotone} implies that $(\bm u^*, \bm 0)$ is globally asymptotically stable. If $(\bm 0, \bm v^*)$ is neutrally stable, by \cite[Theorem 1.4]{LAM2016Munther}, $(\bm u^*, \bm 0)$ is globally asymptotically stable. This proves (ii).
\qed

\noindent\emph{Proof of Theorems \ref{gdyn1} and \ref{gdyn2}.} We only need to prove the case $(d_2,q_2)\in G_{23}^*$, since the other cases can be proved  using similar arguments in the proof of Theorem \ref{gdyn}. If $(d_2,q_2)\in G_{23}^*$, then $q_2<\underline q$ and $(d_1,q_1)\in \hat G_{13}$, where
$$
\hat G_{13}:=\{(d, q):\; d\ge d_2,q_2\le q< \frac{q_2}{d_2}d,(d,q)\ne(d_2,q_2)\}.
$$
Since the nonlinear terms of model \eqref{3p} are symmetric,  it follows from Theorem \ref{gdyn} (ii) that $(\bm 0,\bm v^*)$ is  globally asymptotically stable.

\section*{Appendix}
In the appendix, we study the relations of $\overline q$, $\underline q$ and $q_0$. For convenience, we recall the definition of $\overline q$, $\underline q$ and $q_0$:
\begin{subequations}\label{q0qq2}
\begin{align}
&\overline q=\max\left\{\ds\f{r_1}{k_1}(k_1-k_2),\ds\f{r_3}{k_3}(k_2-k_3)\right\},\label{q0qq2-1}\\
&\underline q=\min\left\{\ds\f{r_1}{k_1}(k_1-k_2),\ds\f{r_3}{k_3}(k_2-k_3)\right\},\label{q0qq2-2}\\
&q_0=\max\left\{r_1\left(1-\frac{u_1^*}{k_1}\right), r_2\left(1-\frac{u_2^*}{k_2}\right)\right\}.\label{q0qq2-3}
\end{align}
\end{subequations}
\begin{lemma}\label{relation}
Suppose that $(\bf {H})$ holds, $\bm r\gg\bm 0$, and $d_1,q_1>0$. Then the following statements hold:
\begin{enumerate}
  \item [\rm (i)] If $q_1<\underline q$, then $q_0>q_1$;
  \item [\rm (ii)] If $q_1>\overline q$, then $q_0<q_1$;
  \item [\rm (iii)] If $q_1>\underline q$, then $q_0>\underline q$;
  \item [\rm (iv)] If $q_1<\overline q$, then $q_0<\overline q$.
\end{enumerate}
\end{lemma}

\begin{proof}
By \eqref{f}-\eqref{sup} and Lemma \ref{uv} (i), we have
\begin{subequations}\label{fpro}
\begin{align}
&\tilde f_1=d_1u_{2}^*-(d_1+q_1)u_1^*=-r_1u_1^*\left(1-\frac{u_1^*}{k_1}\right)<0,\label{fpro-1}\\
&\tilde f_2=d_1u_3^*-(d_2+q_2)u_2^*=r_3\left(1-\frac{u_3^*}{k_3}\right)<0,\label{fpro-2}\\
&\tilde f_2-\tilde f_1=-r_2u_2^*\left(1-\frac{u_2^*}{k_2}\right),\label{fpro-3}
\end{align}
\end{subequations}
which will be used  in the proof below.

(i) By Lemma \ref{uv} (iv), we have $u^*_1>u^*_2>u^*_3$.
This, together with \eqref{q0qq2-3} and \eqref{fpro-1}, implies that
\begin{equation}\label{q0q1}
q_0\ge r_1\left(1-\frac{u_1^*}{k_1}\right)=\f{(d_1+q_1)u^*_1-d_1u^*_2}{u^*_1}=d_1\left(1-\f{u^*_2}{u^*_1}\right)+q_1>q_1.
\end{equation}

(ii) By Lemma \ref{uv} (iii), we have
$u^*_1<u^*_2<u^*_3$.
Then by \eqref{fpro-1} again, we obtain
\begin{equation}\label{r1f1f2}
r_1\left(1-\frac{u_1^*}{k_1}\right)=d_1\left(1-\f{u^*_2}{u^*_1}\right)+q_1<q_1.
\end{equation}
By \eqref{fpro-3}, we obtain that 
\begin{equation}\label{r2f1f2}
\begin{split}
r_2\left(1-\frac{u_2^*}{k_2}\right)=&\ds\f{\left((d_1+q_1)u^*_2-d_1u_3^*\right)+\tilde f_1}{u^*_2}\\
=&d_1\left(1-\f{u^*_3}{u^*_2}\right)+\f{\tilde f_1}{u^*_2}+q_1<q_1,
\end{split}
\end{equation}
where we have used \eqref{fpro-3} and $u_2^*<u_3^*$ in the last step.
It follows from \eqref{q0qq2-3}, \eqref{r1f1f2} and \eqref{r2f1f2} that $q_0<q_1$.

(iii) We divide the proof into three cases:
\begin{equation*}
({\rm A1})\;\;u_1^*<k_2, \;\;({\rm A2})\;\; u^*_1\ge u^*_2,\;\;({\rm A3})\;\; k_2\le u_1^*<u_2^*.
\end{equation*}
For case (A1), we see from \eqref{q0qq2-2} and \eqref{q0qq2-3}
that
\begin{equation*}
q_0\ge r_1\left(1-\frac{u_1^*}{k_1}\right)>\ds\f{r_1}{k_1}(k_1-k_2)\ge\underline q.
\end{equation*}
For case (A2), we see from \eqref{q0qq2-3} and \eqref{fpro-1} that
\begin{equation*}
q_0\ge r_1\left(1-\frac{u_1^*}{k_1}\right)=d_1\left(1-\f{u^*_2}{u^*_1}\right)+q_1\ge q_1>\underline q.
\end{equation*}
Now we consider (A3). Suppose to the contrary that $q_0\le\underline q$. This, combined with \eqref{q0qq2-2}-\eqref{q0qq2-3}, yields
\begin{equation}\label{3cont}
r_1\left(1-\frac{u_1^*}{k_1}\right)\le q_0\le \underline q\le\f{r_3}{k_3}\left(k_2-k_3\right).
\end{equation}
Noticing that $u_2^*>k_2$, we see from \eqref{fpro-3} that
\begin{equation}\label{u1u2u3}
\tilde f_2-\tilde f_1=d_1(u^*_3-u^*_2)-(d_1+q_1)(u^*_2-u^*_1)=-r_2u^*_2\left(1-\ds\frac{u^*_2}{k_2}\right)>0.
\end{equation}
Since $u^*_1< u^*_2$, we see from \eqref{u1u2u3} that $u^*_2< u^*_3$. Then we have
\begin{equation}\label{monou}
0>\tilde f_2>\tilde f_1\;\;\text{and}\;\;k_2\le u^*_1< u^*_2< u^*_3,
\end{equation}
which yields
\begin{equation}\label{f1f2}
-\ds\f{\tilde f_2}{u^*_3}<-\ds\f{\tilde f_1}{u^*_1}.
\end{equation}
This, together with \eqref{monou}, \eqref{fpro-1} and \eqref{fpro-2}, implies that
\begin{equation*}
r_1\left(1-\f{u^*_1}{k_1}\right)=-\ds\f{\tilde f_1}{u^*_1}>-\ds\f{\tilde f_2}{u^*_3}=\ds\f{r_3}{k_3}\left(u^*_3-k_3\right)>\ds\f{r_3}{k_3}\left(k_2-k_3\right),
\end{equation*}
which contradicts \eqref{3cont}. Therefore, $q_0>\underline q$ for case (A3).

(iv) We first show that
\begin{equation}\label{r1}
r_1\left(1-\frac{u_1^*}{k_1}\right)<\overline q,
\end{equation}
and the proof is divided  into three cases:
\begin{equation*}
({\rm B1})\;\;u^*_1> k_2,\;\;({\rm B2})\;\;u_1^*\le u_2^* ,\;\;({\rm B3})\;\; k_2\ge u^*_1> u^*_2.
\end{equation*}
For case (B1), we have
\begin{equation*}
r_1\left(1-\frac{u_1^*}{k_1}\right)<\f{r_1}{k_1}(k_1-k_2)\le\overline q.
\end{equation*}
For case (B2),  we see from \eqref{fpro-1} that
\begin{equation}\label{q0q1-2}
 r_1\left(1-\frac{u_1^*}{k_1}\right)=d_1\left(1-\f{u^*_2}{u^*_1}\right)+q_1\le q_1<\overline q.
\end{equation}
For case (B3), using similar arguments as the above case (A3), we have
\begin{equation*}
0>\tilde f_1>\tilde f_2\;\;\text{and}\;\;k_2\ge u^*_1> u^*_2> u^*_3.
\end{equation*}
This, combined with \eqref{fpro-1} and \eqref{fpro-2}, implies that
\begin{equation*}
r_1\left(1-\frac{u_1^*}{k_1}\right)=-\f{\tilde f_1}{u^*_1}<-\f{\tilde f_2}{u^*_1}<-\f{\tilde f_2}{u^*_3}
=r_3\left(\frac{u_3^*}{k_3}-1\right)<\f{r_3}{k_3}(k_2-k_3)\le\overline q.
\end{equation*}

Then we show that
\begin{equation}\label{r2}
r_2\left(1-\frac{u_2^*}{k_2}\right)<\overline q,
\end{equation}
and the proof is also divided into three cases:
\begin{equation*}
({\rm C1})\;\;u^*_2\le u_3^*,\;\;({\rm C2})\;\;u_2^*> u_3^*\ge k_2 ,\;\;({\rm C3})\;\; u_2^*> u_3^*\;\;\text{and}\;\;k_2>u_3^*.
\end{equation*}
For case (C1), we see from \eqref{r2f1f2} that
\begin{equation*}
r_2\left(1-\frac{u_2^*}{k_2}\right)<q_1<\overline q.
\end{equation*}
For case (C2), we have 
\begin{equation*}
r_2\left(1-\frac{u_2^*}{k_2}\right)<0<\overline q.
\end{equation*}
For case (C3),  we see from \eqref{fpro} that
\begin{equation*}
r_2\left(1-\frac{u_2^*}{k_2}\right)=\f{\tilde f_1-\tilde f_2}{u^*_2}<-\f{\tilde f_2}{u^*_2}<-\f{\tilde f_2}{u^*_3}
=\f{r_3}{k_3}\left(u^*_3-k_3\right)<\f{r_3}{k_3}\left(k_2-k_3\right)\le\overline q.
\end{equation*}
By \eqref{r1} and \eqref{r2}, we see that (iv) holds.
\end{proof}
\begin{remark}
By  $\underline q\le \overline q$ and Lemma \ref{relation}, we see that
 if $q_1< \underline q$, then $q_1<q_0<\overline q$; if $q_1> \overline q$, then $\underline q<q_0<q_1$; and if $\underline q<q_1<\overline q$, then $\underline q<q_0<\overline q$.
\end{remark}
\bibliographystyle{abbrv}
\bibliography{ref}
\end{document}